\documentclass[]{siamart171218}

\usepackage{lipsum}
\usepackage{amsfonts}
\usepackage{graphicx}
\usepackage{epstopdf}
\DeclareGraphicsExtensions{.eps}

\newsiamremark{remark}{Remark}
\newsiamremark{hypothesis}{Hypothesis}
\crefname{hypothesis}{Hypothesis}{Hypotheses}
\newsiamthm{claim}{Claim}

\newcommand{\norm}[1]{\left\lVert#1\right\rVert}
\newcommand{\abs}[1]{\left\lvert#1\right\rvert}

\usepackage{amssymb}
\usepackage{amsmath}
\usepackage{algorithm}
\usepackage[noend]{algpseudocode}

\usepackage{enumitem}

\makeatletter
\def\BState{\State\hskip-\ALG@thistlm}
\makeatother

\usepackage{graphicx}
\graphicspath{ {./} }
\usepackage{caption}
\usepackage{subcaption}

\usepackage{cite}

\def\x{{\mathbf x}}

\def\d{{\mathbf d}}

\def\u{ \mathbf{u}}

\def\y{{\mathbf y}}

\def\z{{\mathbf z}}

\def\u{{\mathbf u}}
\def\U{{\mathcal U}}

\def\D{{\mathbf D}}
\def\d{{\mathbf d}}

\def\A{{\mathbf A}}
\def\K{{\mathbf K}}
\def\e{{\mathbf e}}
\def\B{{\mathbf B}}
\def\W{{\mathbf W}}

\def\b{{\mathbf b}}

\def\s{{\mathbf s}}

\def\E{{\mathbb E}}
\def\I{{\mathbf I}}
\def\0{{\mathbf 0}}
\def\1{{\mathbf 1}}
\def\alfa{{\boldsymbol \alpha}}
\def\gama{{\boldsymbol \gamma}}

\def\O{{\boldsymbol \Omega}}
\def\Phib{{\boldsymbol \Phi}}

\DeclareMathOperator{\Trace}{Trace}
\DeclareMathOperator{\argmin}{argmin}
\DeclareMathOperator{\argmax}{argmax}
\DeclareMathOperator{\st}{s.t.}
\DeclareMathOperator{\rank}{rank}

\headers{Multi-Layer Sparse Coding the Holistic Way}{A. Aberdam, J. Sulam, and M. Elad}

\title{Multi-Layer Sparse Coding: the Holistic Way}

\author{Aviad Aberdam\thanks{Electrical Engineering Department, Technion
Israel Institute of Technology. (\email{ aaberdam@campus.technion.ac.il}).}
\and Jeremias Sulam\thanks{Computer Science Department, Technion Israel Institute of Technology.
(\email{jsulam@cs.technion.ac.il}, \email{ elad@cs.technion.ac.il}).}
\and Michael Elad\footnotemark[3]}

\usepackage{amsopn}

\ifpdf
\hypersetup{
  pdftitle={Multi-Layer Sparse Coding the Holistic Way},
  pdfauthor={A. Aberdam, J. Sulam, and M. Elad}
}
\fi

\begin{document}

\maketitle

\begin{abstract}
The recently proposed multi-layer sparse model has raised insightful connections between sparse representations and convolutional neural networks (CNN). In its original conception, this model was restricted to a cascade of \emph{convolutional synthesis} representations. In this paper, we start by addressing a more general model, revealing interesting ties to fully connected networks. We then show that this multi-layer construction admits a brand new interpretation in a unique symbiosis between synthesis and analysis models: while the deepest layer indeed provides a synthesis representation, the mid-layers decompositions provide an analysis counterpart. This new perspective exposes the suboptimality of previously proposed pursuit approaches, as they do not fully leverage all the information comprised in the model constraints. Armed with this understanding, we address fundamental theoretical issues, revisiting previous analysis and expanding it. Motivated by the limitations of previous algorithms, we then propose an integrated -- \emph{holistic} -- alternative that estimates all representations in the model simultaneously, and analyze all these different schemes under stochastic noise assumptions. Inspired by the synthesis-analysis duality, we further present a Holistic Pursuit algorithm, which alternates between synthesis and analysis sparse coding steps, eventually solving for the entire model as a whole, with provable improved performance. Finally, we present numerical results that demonstrate the practical advantages of our approach.
\end{abstract}



\section{Introduction}

Sparse representation is one of the most popular priors in signal and image processing, leading to remarkable results in various applications \cite{elad2006image,dong2011image,zhang2010discriminative,jiang2011learning,mairal2014sparse}. In its most popular interpretation, this model assumes that a natural signal, represented by the vector $\x \in\mathbb{R}^n$, can be \emph{synthesized} as a linear combination of only a few columns, or atoms, from a matrix $\D\in\mathbb{R}^{n\times m}$, termed a dictionary. In other words, $\x=\D\gama$, where the vector $\gama\in\mathbb{R}^{m}$ is sparse: only a few of its entries are non-zeros, which is indicated by its low $\ell_0$ (pseudo-)norm, $\|\gama\|_0 \ll n$. 

In \cite{bristow2013fast,kavukcuoglu2010fast,heide2015fast,gu2015convolutional,papyan2017working}, this general model was deployed in a Convolutional Sparse Coding (CSC) form, in which the dictionary $\D$ is given by a union of banded and circulant matrices. 
More recently, this CSC model has been extended to a multi-layer version in \cite{papyan2016convolutional}. This construction, termed Multi-Layer Convolutional Sparse Coding (ML-CSC), raises particular interest because of its tight connection to deep learning. Somewhat surprisingly, under this model assumption, the forward pass of a CNN can be interpreted as a pursuit algorithm searching for the respective sparse representations of a given input signal \cite{papyan2016convolutional}. As a result, this provides a promising framework for a theoretical study and analysis of deep learning architectures.

In its original formulation \cite{papyan2016convolutional,sulam2017multi}, this multi-layer model was interpreted as a cascade of synthesis sparse representations. More precisely, every intermediate layer in this model wears two hats: it provides a sparse representation for the previous layer, while also acting as a signal for the subsequent layer. 
This perception has led the authors of \cite{papyan2016convolutional} to propose synthesis-oriented pursuit algorithms that proceed in a layer by layer fashion, propagating the signal from the input to the deepest layer \cite{papyan2016convolutional}. Alternatively, one may adopt a projection approach by seeking for the deepest representation and then propagating it back towards shallower layers \cite{sulam2017multi}. In this paper, we revisit these algorithms more broadly, adopting fully connected layers, providing more general constructions that are not restricted to the convolutional case. 

As we will carefully show, the above pursuit algorithms suffer from several caveats as neither approach can fully leverage all the information in the model. The layer-wise approach provides representations with increasing deviations from the input signal. The projection variant resolves this issue, though it condenses to a traditional global synthesis model that fails to explicitly employ the information represented in the intermediate layers. In addition, the analysis and performance of these methods rely on the intermediate dictionaries being sparse, thus enabling sparse intermediate decompositions. We will show that this does not have to be the case, and one can indeed consider more general dictionaries while still allowing for signals in the model. Alas, the above algorithms collapse in such cases and fail to retrieve the corresponding representations even in noiseless (ideal) cases.

Motivated by these observations, and aiming to effectively resolve these problems, we give the multi-layer sparse model a brand-new interpretation. The key observation is that the ML-SC model provides a unique integration between two types of sparse models: synthesis and analysis \cite{elad2007analysis,selesnick2009signal,nam2011cosparse,nam2013cosparse,rubinstein2013analysis}. As explained above, the synthesis sparse model assumes that a signal $\x$ can be expressed as $\D\gama$, with $\gama$ being sparse. The analysis counterpart -- a somewhat more recent and less glaring variant -- states that a signal $\x$ provides a sparse representation after being multiplied by an analysis dictionary, $\O\in\mathbb{R}^{m\times n}$. In this way, $\|\O\x\|_0 = m-\ell$, where the number of zeros, $\ell$, refers to the cardinality of the \emph{cosupport} of $\x$, termed \emph{cosparsity}. With this understanding, we will show that while the outer or global shell of the ML-SC model provides a synthesis representation for a given signal, the intermediate representations enforce an analysis prior on the deepest representation. 
 
This new synthesis-analysis conception leads us to re-formulate and expand on several theoretical questions, such as: When is the model valid? More precisely, are there signals that admit all model constraints? What is the ML-SC signal space? How can we synthesize an ML-SC signal under this interpretation? How can uniqueness guarantees be improved based on these extra constraints? The answers to these questions will shed light on the achievable sparsity bounds of the intermediate representations. Interestingly, our analysis shows that these should, in fact, not be \emph{too sparse}. This will free the restriction of employing sparse matrices as the intermediate dictionaries, on the one hand, and provide a closer behavior to what is observed in practical deep neural networks, on the other. 

Driven by the the limitations of previous pursuit algorithms and the offered analysis-synthesis perspective, we turn to define a novel pursuit approach. Our proposed method seeks to estimate all representations in the model simultaneously, for which we dubbed it a \emph{Holistic} pursuit. We first analyze the performance of its \emph{oracle} estimator (i.e., having knowledge of the underlying representation supports), and compare it with the corresponding estimators of the previous layer-wise and projection alternatives. 

We then propose a practical Holistic Pursuit algorithm, which iteratively builds on the intermediate layer supports while refining the global, synthesis, representation -- improving on it at every step. To this end, this algorithm alternates between a synthesis-type sparse coding of the deepest layer, and an analysis-like estimation of the mid-layers supports using the deepest layer estimation. This holistic approach leverages the synergy in the across different layers, resulting in improved provable recovery guarantees and performance bounds.

This paper is organized as follows: In \Cref{Sec:Background} we review the basics of the sparse representations model. \Cref{Sec:BackgroundMLSC} then introduces previous theoretical claims for the ML-CSC model and adapts them to a more general (non-convolutional) case.
In \Cref{Sec:Synthesis_Interpretation_Issues} we demonstrate the limitations of the existing multi-layer synthesis interpretation, and introduce an analysis perspective to these constructions. We undertake the study of uniqueness guarantees in light of the synthesis-analysis duality in \Cref{Sec:Uniqueness}, and present a holistic approach for the pursuit of these representations.
\Cref{Sec:Oracle} provides the analysis of the Oracle estimators for the different pursuit approaches under random noise assumption, while in \Cref{Sec:HolisticPursuit} we present a new pursuit algorithm for the ML-SC model, the \emph{Holistic Pursuit}, which we demonstrate with numerical experiments in \Cref{Sec:Experiments}. We finally conclude in \Cref{Sec:conclusions}, delineating further research directions.

\subsection{Notations}
Vectors and matrices are denoted by bold lower and upper case letters, respectively. 
$\x^{\Lambda}$ denotes the vector in $\mathbb{R}^{\abs{\Lambda}}$ that carries the entries of $\x$ indexed by $\Lambda$. Naturally, $\Lambda^c$ will denote the complement of the set $\Lambda$. Similarly, when $\D$ is a matrix in $\mathbb{R}^{n \times m}$ and $\Lambda_c \subseteq \{1, \ldots,m\}$, $\D^{\Lambda_c}$ is the sub-matrix in $\mathbb{R}^{n \times \abs{\Lambda_c}}$ whose columns are those of $\D$ indexed by $\Lambda_c$. 
If $\Lambda_r \subseteq \{1, \ldots,n\}$, then the matrix $\D^{\Lambda_r, \Lambda_c}$ is the matrix in $\mathbb{R}^{\abs{\Lambda_r} \times \abs{\Lambda_c}}$ whose rows are those of $\D^{\Lambda_c}$ indexed by $\Lambda_r$. We will further denote by $\D^{\Lambda_r,\mathcal{I}}$ the matrix containing the rows indexed by $\Lambda_r$ across all $m$ columns of $\D$, where $\mathcal{I}$ denotes the index set $\mathcal{I} = [1,m]$.


\section{Sparse Representation Modeling Background}
\label{Sec:Background}

Many natural images and signals have been observed to be inherently low dimensional despite their possibly very high ambient dimension. Sparse representations offers an elegant and clear way to model such inherent low-dimensionality by assuming that the signal $\x \in \mathbb{R}^{n}$ belongs to a finite (yet, huge) union of $s$ ($\ll n$) dimensional subspaces \cite{lu2008theory}. This general idea comes in two forms: the traditional and very popular synthesis approach \cite{elad2010exact}, and the newer and complementary analysis sparse model \cite{nam2011cosparse,nam2013cosparse}, which we review next.

\subsection{The Synthesis Model}

Synthesis sparse representations is a signal model that assumes that natural signals can be represented, or well approximated, by a linear combination of a few basic components, termed atoms. Formally, such a signal $\x \in \mathbb{R}^{n}$ can be expressed as $\x = \D\gama$, where $\D \in \mathbb{R}^{n\times m}$ is a dictionary containing signal atoms as its columns, and the vector $\gama \in \mathbb{R}^{m}$ contains only a few non-zero entries. The cardinality of a vector is measured by the $\ell_0$ pseudo-norm, $\|\gama\|_0$.
Typically, we are interested in the case of redundant dictionaries, i.e. $m>n$, allowing for very sparse representations.

The synthesis inverse problem aims to recover the sparse representation $\gama$ from a noisy signal observation $\y=\x+\e=\D\gama+\e$, where $\e$ is a noise vector and the dictionary $\D$ is assumed given. This task is often called sparse coding, or simply pursuit, and can be formally written\footnote{For the sake of simplicity and to avoid introducing more notation, we will employ hereafter the same variable to denote the ground truth and the running variable we are optimizing over - these are not to be mixed up.} as \cite{donoho2003optimally,tropp2004greed,elad2010exact}:
\begin{equation} \label{P_0_epsilon_definition}
\left( P_0 \right): ~ \min_{\gama} \norm{\gama}_0~ \textrm{s.t.} ~ \norm{\D \gama - \y}_2 \leq \epsilon .
\end{equation}
Since solving the problem in \Cref{P_0_epsilon_definition} is an NP-hard in general \cite{gribonval2003sparse}, one can use greedy strategies like Orthogonal Matching Pursuit (OMP) \cite{pati1993orthogonal} or the thresholding algorithm \cite{tropp2004greed,elad2010exact} to approximate its solution.
Alternatively, one can also use a convex relaxation of this pursuit by replacing the $\ell_0$ norm with the convex $\ell_1$. In the latter case, the resulting problem, termed Basis-Pursuit, is defined formally as \cite{chen2001atomic,donoho2003optimally,tropp2006just}:
\begin{equation}
\left( P_1 \right): ~ \min_{\gama} \norm{\gama}_1~ \textrm{s.t.}~ \norm{\D \gama - \y}_2 \leq \epsilon .
\end{equation}

In the noiseless case, where $\epsilon=0$, one of the fundamental questions is whether and when one can be sure that the result of these approximation algorithms is in fact the unique sparsest representation of the signal. Equivalently, from a \emph{transformation} perspective, these questions explore the conditions under which the sparse synthesis operation is invertible.
A key property for the study of uniqueness is the \emph{spark} of the dictionary,  $\sigma(\D)$, defined as the smallest number of columns from $\D$ that are linearly-dependent.
If there exists a representation $\gama$ for a signal $\x$ such that $\norm{\gama}_0 < \frac{\sigma(\D)}{2}$,
then this solution is necessarily the sparsest possible \cite{donoho2003optimally}; { in other words, such a condition is sufficient for the representation to be unique.}
However, the spark is at least as difficult to evaluate as solving $\left( P_0 \right)$, and thus it is common to lower bound it with the mutual-coherence, $\mu(\D)$. This value is simply the maximal correlation between atoms in the dictionary:
\begin{equation}
\mu(\D) = \max_{i\ne j} \frac{\abs{\d_i^T \d_j}}{\norm{\d_i}_2 \norm{\d_j}_2 }
,
\end{equation}
where we have denoted by $\d_j$ the $j^{th}$ column of the matrix $\D$.
One may then bound the spark with the mutual coherence \cite{donoho2003optimally}, as $\sigma(\D) \geq 1 + \frac{1}{\mu(\D)}$. Then, a sufficient condition for uniqueness is to require that
\begin{equation}
\norm{\gama}_0 < \frac{1}{2} \left( 1 + \frac{1}{\mu(\D)} \right).
\end{equation}

\subsection{The Analysis Model}

The above model has an analysis counterpart, in which the assumption is that one can linearly transform the signal into a sparse representation \cite{nam2011cosparse,nam2013cosparse}.
Formally, for a fixed analysis operator $\O \in \mathbb{R}^{m\times n}$, a signal $\x \in \mathbb{R}^n$ belongs to the analysis model with co-sparsity $\ell$ if $\norm{\O\x}_0 = m - \ell$.
When $\O$ is a squared unitary matrix, the analysis model is identical to the synthesis one with dictionary $\D = \O^T$. However, in the overcomplete case where $m > n$, there is no simple connection between the two as they lead to generally different constructions.
Note that, unlike the synthesis counterpart, the pursuit in the analysis model is trivial in the noiseless case as it simply amounts to a multiplication by the analysis dictionary $\O$.
In the noisy case, the process of recovering $\x$ from the corrupted measurements $\y = \x + \e$ is done by solving the following minimization problem \cite{elad2007analysis}:
\begin{equation} \label{P_0_ell}
\left( P_0^\ell \right): ~ \min_{\x} \norm{\O \x}_0~ \textrm{s.t.} ~ \norm{\x - \y}_2 \leq \epsilon.
\end{equation}
Just as the $\left( P_0 \right)$ problem, this objective is NP-hard in general, and so one must resort to greedy approaches \cite{giryes2014greedy,nam2013cosparse} or $\ell_1$ relaxation alternatives \cite{elad2007analysis,candes2011compressed,nam2013cosparse}.

Lastly, a third type of sparse model is that of Sparsifying Transforms \cite{ravishankar2013learning,ravishankar2015online}. This analysis-type model seeks for a (typically square) dictionary $\W$ -- the transform -- that approximately sparsifies a signal $\y$, so that $\W\y = \gama + \e$, where $\norm{\gama}_0 \ll n$ and $\e$ is some nuisance (dense) vector. The optimization problems related to this model present interesting advantages, as the pursuit is no longer an NP-hard problem but rather a simple thresholding operation. We will not duel on Transform Learning any further in this paper, but we believe that many of the ideas raised in our work could be adapted to this model form as well.

\section{The Multi-Layer Sparse Coding Model}
\label{Sec:BackgroundMLSC}
While the above sparse models have been around for nearly two decades, a multi-layer extension was only recently introduced. This was done in a convolutional setting, thus termed Multi-Layer Convolutional Sparse Coding (ML-CSC) \cite{papyan2016convolutional,sulam2017multi}. This model is an extension of the convolutional sparse model \cite{bristow2013fast,papyan2017working}, which addresses the modeling of high dimensional signals through local shift-invariant sparse decompositions.
It is our intention in this work to consider a more general case and not to restrict ourselves to the convolutional scenario. We refer the interested reader to \cite{papyan2017working,papyan2017slicedbased} for a thorough review of convolutional sparse representations, their associated results and algorithms.

\subsection{Model and Pursuit Definitions}
The synthesis sparse model assumes that a signal $\x \in \mathbb{R}^n$ can be decomposed into a multiplication of a dictionary $\D_1 \in \mathbb{R}^{n\times m_1}$ and a sparse vector $\gama_1 \in \mathbb{R}^{m_1}$.
In the multi-layer model we extend this by assuming that $\gama_1$, and in fact every sparse representation, $\gama_i$, can also be decomposed as $\gama_i =\D_{i+1}\gama_{i+1}$, where $\D_{i+1} \in \mathbb{R}^{m_i\times m_{i+1}}$ is the dictionary of layer $i+1$ and $\gama_{i+1} \in \mathbb{R}^{m_{i+1}}$ is the corresponding sparse representation. We name this the Multi-Layer Sparse Coding (ML-SC) model, and formalize its definition as follows.

\begin{definition} (ML-SC signal):
Given a set of dictionaries ${\{\D_i}\}_{i=1}^k$, of appropriate dimensions, a signal $\x\in \mathbb{R}^n$ admits
a representation in terms of the ML-SC model if
\begin{equation}\begin{array}{rclc} 
\x  &= & \D_1\gama_1, & \norm{\gama_1}_0 \leq s_1 , \\ 
\gama_1 & = & \D_2\gama_2,  & \norm{\gama_2}_0 \leq s_2 ,\\
  & \vdots   \\
\gama_{k-1} & = & \D_k\gama_k,  & \norm{\gama_k}_0 \leq s_k .\\
\end{array}\end{equation}
\end{definition}
For the purpose of the following derivations, define $\D_{(i,j)}$ to be the effective dictionary from the $i^{th}$ to the $j^{th}$ layer, i.e., $\D_{(i,j)} = \D_i \D_{i+1} \cdots \D_j$. This way, one can concisely write $\gama_i = \D_{(i,j)} \gama_j$. 
For effective dictionaries from the first layer to the $j^{th}$ level, we simplify the notation and denote $\D_{(j)} = \D_{(1,j)}$, so that $\x = \D_{(i)} \gama_i$.
The ML-SC can then be interpreted as a global synthesis model, $\x = \D_{(k)}\gama_k$, with additional intermediate layer constraints.
As we will see, these two observations are the laying foundation for the current pursuit algorithms proposed by recent works (\cite{papyan2016convolutional,sulam2017multi}).
We now formalize the pursuit for the ML-SC model, referred to as Deep Pursuit:
\begin{definition} (Deep Pursuit):
For a signal $\y=\x+\e$, where $\x$ is an ML-CS signal and $\e$ is an additive noise, assume that the set of dictionaries ${\{\D_i}\}_{i=1}^k$, the cardinality vector $\s$, and the noise energy $\epsilon$, are all known. Define the Deep Pursuit $(DP_{\s})$ problem as:
\begin{equation}
\begin{split}
\left( DP_{\s} \right): ~ \textrm{find} ~ {\{\gama_i\}}_{i=1}^k ~~ \textrm{s.t.} ~ & \norm{\y- \D_1\gama_1}_2 \leq \epsilon \\
				& \gama_{i-1} = \D_i \gama_i, ~\forall ~ 2\leq i \leq k \\
                & \norm{\gama_i}_0 \leq s_i , ~\forall ~ 1\leq i \leq k.
\end{split}
\end{equation}
where the scalar $s_i$ is the $i^{th}$ entry of $\s$.
\end{definition}

We are interested in the following questions: Is the solution of $\left( DP_{\s} \right)$ unique in a noiseless ($\epsilon = 0$) setting? Is the solution stable to noise contamination? And under which conditions would these be true? 

\subsection{Uniqueness}
Consider a set of dictionaries $\{\D_i\}_{i=1}^k$ and a signal $\x$  admitting a multi-layer sparse representation defined by the set $\{\gama_i \}_{i=1}^k$. The claim of uniqueness answers the question of whether another set of sparse vectors can represent the same signal $\x$.
We present here the uniqueness theorem from \cite{papyan2016convolutional}, with necessary changes that make it suitable to the general (i.e., non-convolutional) multi-layer sparse model. 

\begin{theorem} 
(Uniqueness via the mutual coherence): Consider a noiseless ML-SC signal $\x$, its set of dictionaries $\left\{ \D_i \right\}_{i=1}^k$ and their corresponding mutual coherence constants $ \mu(\D_i),~\forall 1\le i \le k $. If
\begin{equation} \label{uniqueness_theorem_condition}
\forall ~ 1 \leq i \leq k, ~~ \norm{\gama_i}_0 = s_i < \frac{1}{2} \left( 1 + \frac{1}{\mu(\D_{i})} \right),
\end{equation}
then the set $\left\{ \gama_i \right\}_{i=1}^k$ is the unique solution to the $DP_\s$ problem.
\end{theorem}

The simple \emph{modus operandi} behind this result is to propagate the uniqueness guarantees progressively through the layers. In other words, it first demands the first layer to be the unique representation of the signal, then it demands the second layer to be the unique representation of the first layer, and so on.
In \cite{sulam2017multi}, the authors suggested an improvement based on a projection approach. Instead of propagating the uniqueness conditions layer by layer, one can project the signal directly to the deepest representation layer, and to demand uniqueness using the effective model, $\D_{(k)}$, requiring:
\begin{equation} \label{uniqueness_theorem_condition_effective_dictionary}
\norm{\gama_k}_0 < \frac{1}{2}\left( 1 + \frac{1}{\mu(\D_{(k)})} \right). 
\end{equation}

An immediate way of improving over these uniqueness conditions is to maximize between \Cref{uniqueness_theorem_condition} and \Cref{uniqueness_theorem_condition_effective_dictionary}. One can leverage this idea in order to maximize over all the combination of the mid-effective dictionaries, $\D_{(i,j)}$, and if there is any partition of $\{1,\ldots,k\}$ such that $\gama_k$ is guaranteed to be unique then all the representation set $\{\gama_i \}_{i=1}^k$ is thus also unique. However, as we will see in the next section, all these variants of uniqueness guarantees are in fact too restrictive, and do not capture the true essence of the ML-SC model. In \Cref{Sec:Uniqueness} we will revisit this matter and provide a better study of this property, with far tighter bounds.

\subsection{Stability}
Real signals might contain noise or deviations from the idealistic model assumptions presented above. In these cases, one would like to know what the error in the estimated representation is, and how sensitive the different pursuit formulations are to different levels of noise. In other words, we would like to analyze the  stability of the solutions to the pursuit problems. The theorem below is an adaptation of a result from \cite{papyan2016convolutional}.
\begin{theorem}
(Stability of the $DP_\s$ problem): Suppose an ML-SC signal $\x$ is contaminated with energy-bounded noise $\e$, $\norm{\e}_2 \leq \epsilon$, resulting in $\y = \x+\e$, and suppose the set of solutions $\{ \hat{\gama}_i \}_{i=1}^k$ is obtained by solving the $DP_\s$ problem. If the true representation set satisfies the uniqueness conditions in \Cref{uniqueness_theorem_condition},
then
\begin{equation}
\left\Vert \gama_i - \hat{\gama}_i \right\Vert_2^2 \leq 
4 \epsilon^2  \prod_{j=1}^{i} 
{\frac{1}{1-\left(2s_j-1 \right) \mu (\D_{ j } )}}.
\end{equation}
\end{theorem}
\noindent
Clearly, one could use the same improvements suggested above regarding the uniqueness analysis in order to further strengthen this result. 

Note that the above result refers to the solution of the $DP_\s$ problem -- though without specifying how such solutions could be estimated in practice. We now turn to address this aspect, and present the existing pursuit algorithms as introduced in \cite{papyan2016convolutional} and later in \cite{sulam2017multi}, which aim to solve the $DP_\s$ problem.

\subsection{The Layer by Layer Pursuit}
The first method proposed for solving the $DP_\s$ problem was presented in \cite{papyan2016convolutional}, where the idea is to approximately solve each layer progressively, propagating the solution from the first layer to the deeper ones.
In \Cref{Layered_Thresholding} we present the Layered Pursuit algorithm, introduced and theoretically analyzed in \cite{papyan2016convolutional}, and which was shown to be connected to the forward pass of neural networks.
Note that two such variants were suggested in \cite{papyan2016convolutional}: one relying on the Thresholding algorithm, and the other on the Basis Pursuit alternative. \Cref{Layered_Thresholding} presents these two options together, where we denote by $\mathcal{H}(\cdot)$ a thresholding operator, and $P_1(\D,\y,\lambda) \triangleq
\underset{\gama}{\argmin} \norm{\D \gama  - \y }_2^2 \st \|\gama\|_1 \leq \lambda$.

\begin{minipage}{.45\textwidth}
\begin{algorithm} [H]
\caption{The Layered Pursuit algorithm} \label{Layered_Thresholding}
\textbf{Input} \\
\hspace*{\algorithmicindent} $\y$ - a signal.\\
\hspace*{\algorithmicindent} $\left\{ \D_i \right\}_{i=1}^k$ - a set of dictionaries.\\
\textbf{Output} \\
\hspace*{\algorithmicindent} $\{ \hat{\gama}_i \}_{i=1}^k$ - a set of representations.\\
\textbf{Process}
\begin{algorithmic}[1]
\State $\hat{\gama}_0 \gets \y$
\For {$i = 1:k$} \footnotesize
\State $\hat{\gama}_i = \left\{ \begin{array}{cc}
\mathcal{H}\left(\D_i^T \hat{\gama}_{i-1}\right) & \textrm{Thrs.}  \\
P_1(\D_i,\hat{\gama}_{i-1},\lambda_i) & \textrm{BP}
\end{array} \right.$
\EndFor
\State \Return $\{ \hat{\gama}_i \}_{i=1}^k$
\end{algorithmic}
\end{algorithm}
\vspace{.2cm}
\end{minipage} \hfil
\begin{minipage}{.45\textwidth}
\begin{algorithm} [H]
\caption{The Basic Projection Pursuit algorithm} \label{Projection_Thresholding}
\textbf{Input} \\
\hspace*{\algorithmicindent} $\y$ - a signal.\\
\hspace*{\algorithmicindent} $\left\{ \D_i \right\}_{i=1}^k$ - a set of dictionaries.\\
\textbf{Output} \\
\hspace*{\algorithmicindent} $\{ \hat{\gama}_i \}_{i=1}^k$ - a set of representations.\\
\textbf{Process}
\begin{algorithmic}[1] \footnotesize
\State $\hat{\gama}_k = \left\{ \begin{array}{cc}
\mathcal{H}\left(\D_{(k)}^T \y\right) & \textrm{Thres.} \\
P_1(\D_{(k)},\y,\lambda_k) & \textrm{BP}
\end{array} \right.$
\For {$i = k-1: -1: 1$}
\State $\hat{\gama}_i \gets \D_{i+1} \gama_{i+1}$
\EndFor
\State \Return $\{ \hat{\gama}_i \}_{i=1}^k$
\end{algorithmic}
\end{algorithm}
\vspace{.2cm}
\end{minipage}

It is important to note that, while providing approximations, this algorithm does not recover a valid ML-SC signal as it only guarantees that $\gama_{i-1} \approx \D_i \gama_i$.
Another drawback is the recovery error which grows as a function of the network's depth, contradicting the intuition that additional information should decrease the error. 

\subsection{The Projection Pursuit}
An alternative to the layered pursuit is the projection approach presented in \cite{sulam2017multi}. In this algorithm, one first finds the deepest representation using the effective dictionary $\D_{(k)}$, and then propagates this solution all the way back to the first layer.
In \Cref{Projection_Thresholding} we present the simplified version of the Projection Pursuit algorithm, noting that in \cite{sulam2017multi} the authors suggested an improvement that iteratively backtracks if the propagated mid-layer representations violate the model constraints, and attempts to find an alternative sparser and feasible representation, $\gama_k$, in a greedy manner.

This algorithm, if successful, provides an estimation which, unlike the previous case, satisfies the ML-SC constraints. Similar to the behavior for the uniqueness guarantees appear in \Cref{uniqueness_theorem_condition_effective_dictionary}, this approach provides a looser condition and, as presented in \cite{sulam2017multi}, the recovery error is reduced.
However, this algorithm is essentially a single-layer effective model and therefore it does not explicitly use all the available information.

\section{The Synthesis-Analysis Interpretation}
\label{Sec:Synthesis_Interpretation_Issues}

So far, the multi-layer model was interpreted as an extension of the general synthesis model. We present here several concerns that follow from this understanding:
\begin{enumerate}[wide, labelwidth=!]
\item Sparse dictionaries: The approaches presented above had no choice but to enforce sparsity on the intermediate dictionaries in order to ensure sparse intermediate representations. In the more general case, where the intermediate dictionaries are dense, the two presented algorithms simply fail: The layer-wise approach would cause a very high error since dense dictionaries do not span well sparse signals. As a result, every mid-layer pursuit, except from the first one, would result in a very low SNR which further decreases as we go deeper into the model layers. The projection alternative, on the other hand, would converge to the zero representation, because even if it would attain a reasonable deepest layer estimation, the corresponding intermediate representations would become fully dense due to the estimation noise. Following the backtracking in the proposed algorithm, 
the deepest layer cardinality would have to reduce in an attempt to decrease the intermediate cardinalities, eventually resulting in zero representations.
\item Spanned space:
Under the current interpretation, it is unclear what is the space spanned by the ML-CS model. This is related to the following questions:
\begin{itemize}
  \item Empty model -- Are there signals, and their corresponding representations, that satisfy the model constraints, or is the model empty?
  \item Model sampling -- If the model is not empty, how can we synthesize signals satisfying the model constraints? What are the restrictions on the model parameters?
  \end{itemize}
\item Recovery error: The ML-SC signal belongs to a model that is far more constrained than the single-layer version, as additional conditions are introduced in the form of the sparsity of the intermediate representations. Correspondingly, it is expected that the recovery error will be significantly better given these increased constraints. As we have shown above, however, the error in the layer-wise method increases across the layers, and the projection method provides error bounds that are basically \emph{single-layer} type estimates. This indicates that the current approaches provide sub-optimal solutions in attempting to solve the multi-layer pursuit problem.
\end{enumerate}

\subsection{The Synthesis-Analysis Interpretation} \label{Subsection_Synthesis_Analysis_interpretation}
Motivated by these unsolved issues, we now propose to look at the ML-SC model as a one piece rather than as a collection of single-layer constructions. We interpret this model as a unique combination between the synthesis and the analysis paradigms. While the outer shell maintains a synthesis interpretation,
$\x = \D_{(k)} \gama_k, ~ \norm{\gama_k}_0 \leq s_k$,
the intermediate constraints can be understood as analysis constraints on the deepest representation: $\norm{\gama_i}_0 = \norm{\D_{(i+1,k)} \gama_k}_0 \leq s_i, ~~ \forall 1\leq i < k$.

Armed with this observation we can begin re-examining the ML-SC model. 
First, we would like to identify the true space spanned by the signals satisfying the model constraints.
Relying on the effective dictionary (synthesis) interpretation, we know that these signals lie in a union of subspaces composed of all the options of choosing $s_k$ columns from $\D_{(k)}$, each spanning a subspace of dimension $s_k$. 
However, invoking the intermediate analysis constraints must influence further this construction. For known co-supports $\Lambda_1^c, \Lambda_2^c, \ldots, \Lambda_{k-1}^c$, and support $\Lambda_k$, we define the following two matrices: 
\begin{equation} \label{phi_matrix}
\Phib \triangleq 
\renewcommand\arraystretch{2}
\begin{bmatrix}
\D_{(2,k)}^{\Lambda_1^c,\mathcal{I}} \\
\D_{(3,k)}^{\Lambda_2^c,\mathcal{I}} \\
\vdots \\
\D_{(k,k)}^{\Lambda_{k-1}^c,\mathcal{I}}
\end{bmatrix}
\textrm{ and } 
\Phib^{\Lambda_k} \triangleq 
\renewcommand\arraystretch{2}
\begin{bmatrix}
\D_{(2,k)}^{\Lambda_1^c, \Lambda_k} \\
\D_{(3,k)}^{\Lambda_2^c, \Lambda_k} \\
\vdots \\
\D_{(k,k)}^{\Lambda_{k-1}^c, \Lambda_k}
\end{bmatrix}
. \end{equation}
The rows in $\Phib$ define directions to which $\gama_k$ must be orthogonal, as they refer to the zeros in all the intermediate representations. Clearly, the definition of the above two matrices depends not only on the support of $\gama_k$, but also on the co-supports of the intermediate representations.

Considering the fact that $\gama_k$ is $s_k$ sparse and its support is $\Lambda_k$, the matrix $\Phib^{\Lambda_k}$ defines a null-space to which the non-zeros in $\gama_k$ belong. 
Thus, the degrees of freedom in choosing the deepest representation reduces from $s_k$ to $s_k-\rank\{\Phib^{\Lambda_k}\}$. In other words, these signals no longer live in a union of sub-spaces of dimension $s_k$, but rather in a union of $s_k-\rank\{\Phib^{\Lambda_k}\}$ dimensional sub-spaces. Denoting by $\ell_i$ the co-sparsity of $\gama_i$, the number of such sub-spaces hence grows from $\begin{pmatrix} m_k \\ s_k \end{pmatrix}$ to
$
\begin{pmatrix} m_k \\ s_k \end{pmatrix}
\prod_{i=1}^{k-1}
\begin{pmatrix} m_i \\ \ell_i  \end{pmatrix}
$, which is the number of the supports options. These elements in the ML-SC model are depicted in \Cref{fig:mlsc_model} for a two-layer model.

\begin{figure*} \centering
\includegraphics[width=.9\textwidth]{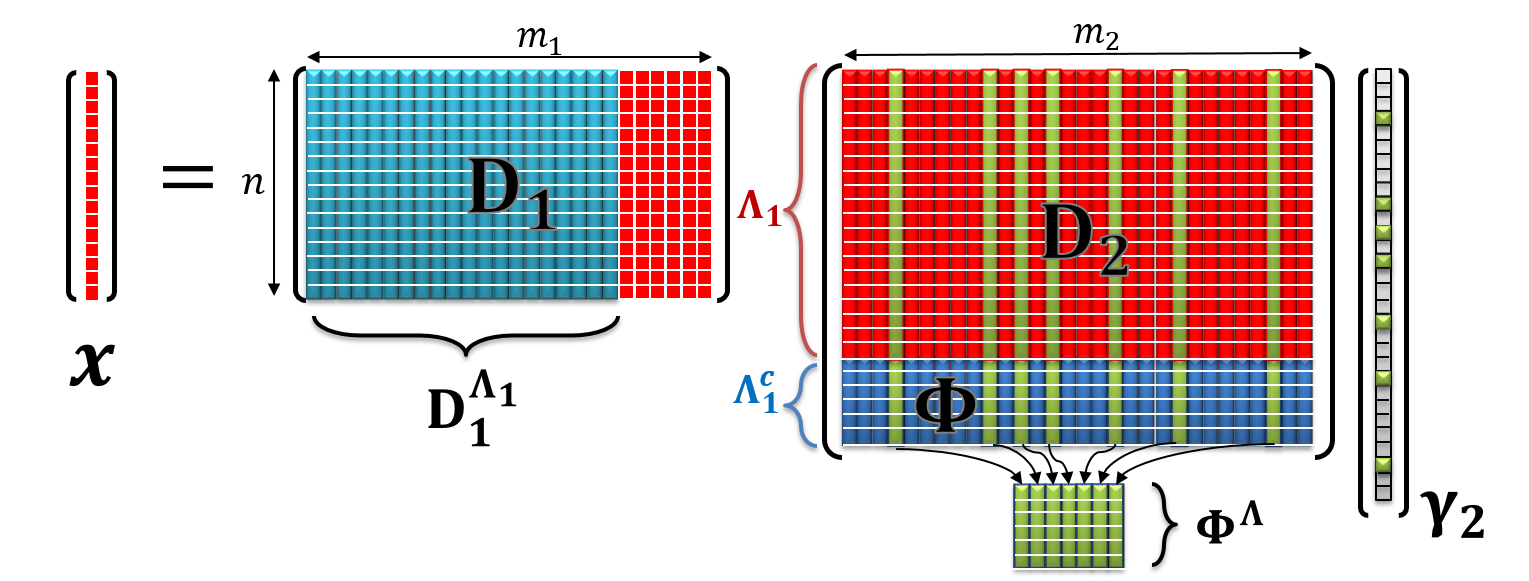}
\caption{Illustration of the ML-SC model for a two-layer decomposition. \vspace{-.6cm}
}
\label{fig:mlsc_model}
\end{figure*}

Several theoretical corollaries can be derived from this analysis, answering some of the questions and issues posed in the previous section:
\begin{enumerate}[wide, labelwidth=!]
\item Empty model: an immediate corollary is that as long as $ s_k > \rank\{\Phib^{\Lambda_k}\} $, the model is not empty, because using the rank-nullity theorem, we know that 
\begin{equation}
\dim\left\{\ker\left(\Phib^{\Lambda_k}\right)\right\} = 
s_k - \rank\{\Phib^{\Lambda_k}\}
.\end{equation}
Then, as long as $ s_k > \rank\{\Phib^{\Lambda_k}\} $, there exists a $\gama_k$ that satisfies the model constraints. 

\item Not \emph{so sparse} intermediate layers: a very interesting benefit is that the mid-layer representations must not need to be too sparse. In fact, contrary to previous works that limited the mid-layers cardinality, in the analysis-synthesis view one limits the intermediate co-cardinality, i.e., the number of zeros, such that $ s_k > \rank\{\Phib^{\Lambda_k}\} $. 
This property mimics the typical behavior of deep neural networks, in which the intermediate representations (or activations \cite{papyan2016convolutional}) are sparse but not extremely so.

\item Model sampling: one can now devise a systematic way to sample a signal from the model. The first step is to randomly choose the representations support, or equivalently, select one of the subspaces. Then, one can multiply the matrix $\K$, which spans the null space of $\Phib^{\Lambda_k}$, with a random vector $\alfa \in \mathbb{R}^{ s_k-\rank\{\Phib^{\Lambda_k}\}}$, resulting in the non-zero coefficient in $\gama_k = \K \alfa$. Finally, the multiplication of $\gama_k$ by the effective dictionary, $\D_{(k)}$, produces the desired ML-SC signal.

\item Sparse dictionaries: recall the need of previous works to consider intermediate sparse dictionaries. Such a construction can be understood only as a particular case of this model, where the obtained representations are indeed sparse (due to the sparse atoms) but not because non-trivial orthogonality was enforced between the deepest representation and the intermediate dictionaries. This way, the matrix $\Phib^{\Lambda_k}$ becomes the zero matrix almost surely, and therefore, the signal lies in a subspace of dimension $s_k$ and the intermediate constraints are passive. 
One might think that in this case there is no extra advantage in the multi-layer model over the single-layer one. However, the matrix $\Phib$ \emph{does} contain information on $\gama_k$ and therefore it is expected to improve the recovery of the support (as in the correcting-support version of the projection algorithm). This can be understood by enforcing the constraint that every new non-zero that is to be added to $\gama_k$ should be orthogonal to the matrix $\Phib$. Indeed, this will be exploited by the algorithm presented in Section \ref{Sec:HolisticPursuit} and its benefits will become a lot clearer then. As we see, the matrix $\Phib$ has two functions: aiding the detection of the true support, for which one should employ the whole matrix $\Phib$, and estimating the values in $\gama_{\Lambda_k}$, for which the sub matrix $\Phib^{\Lambda_k}$ is the one of interest.

\item Random dictionaries: in the other extreme, when the dictionaries are fully dense and sampled from a continuous distribution, the rank of $\Phib^{\Lambda_k}$ is equal to $\sum_{i=1}^{k-1} \ell_i$ with probability 1. Therefore, there are $s_k - \sum_{i=1}^{k-1} \ell_i$ degrees of freedom in choosing $\gama_k$, implying that the signal dimension is significantly reduced. 

\item Recovery error: projecting the signal on a smaller dimensional space reduces the recovery error, and therefore, the ML-SC model is expected to give a significant improvement over the single-layer model. We will extend on this matter later, but we anticipate that this error is proportional to the degrees of freedom, enabling a significant improvement in the ML-SC model.

\item A Holistic alternative: The new interpretation points to the fact that the various representations in all the layers should be estimated jointly as opposed to (relatively) independently or sequentially. This will motivate the derivation of our Holistic Pursuit, to be presented in \Cref{Sec:HolisticPursuit}.
\end{enumerate}

\section{Uniqueness revisited}
\label{Sec:Uniqueness}
The new analysis perspective motivates us to derive a new uniqueness theorem. The following result reflects the underlying benefits of the ML-SC model, exemplifying the gain one can obtain in return for more constrained assumptions. In the proof below we combine the spark -- a synthesis characterization \cite{donoho2003optimally,bruckstein2009sparse} -- with the union of subspaces interpretation \cite{nam2013cosparse,lu2008theory}. In addition, we will require the dictionaries to be in \emph{general position}, as in \cite{nam2013cosparse}, and we defer the precise definition of this characterization to the \Cref{app:UniquenessRevProof}.

\begin{theorem} 
Consider an ML-SC signal $\x$, and a set of dictionaries $\{\D_i\}_{i=1}^{k}$ in general position. 
If there exists a set of representations, $\{\gama_i \}_{i=1}^k$ satisfying
\begin{equation}
s_k \leq \frac{\sigma(\D_{(k)}) -1 }{2} + r,
\end{equation}
where $r = \rank\{ \Phib^{\Lambda_k} \}$ and $s_k$ is number of non-zero coefficients in the deepest layer, then this set is the \emph{unique} ML-SC representation for $\x$ such that its deepest layer has no more than $s_k$ non-zeros and the rank of the corresponding $\Phib^{\Lambda_k}$ is no greater than $r$.
\end{theorem}

Before moving forward, a comment is in place. The traditional uniqueness claims in \cite{papyan2016convolutional,sulam2017multi} provide uniqueness guarantees only when $s_k \leq \frac{\sigma(\D_{(k)}) -1 }{2}$,
since such results use only the synthesis-type interpretation.
Now, the above result efficiently leverages the analysis prior imposed on $\gama_k$, resulting in less restrictive conditions for uniqueness. For the sake of brevity, we defer the proof to \Cref{app:UniquenessRevProof}. 

Finally, note that the above guarantees exclude dictionaries that are not in general position. In this way, some relevant and practical dictionaries, such as certain wavelet frames and total variation, are not covered by the above theorem. We believe that this result can in fact be extended to consider linear dependences and dictionaries not in general position, and this might be studied in detail in future work.

\section{The Oracle Estimator}
\label{Sec:Oracle}
The above result begins to show the benefit of considering all representations simultaneously. In this section, we aim to quantify this precisely by analyzing the performance of the oracle estimator: suppose one knows the true supports across all layers, what is then the optimal (oracle) estimator for the all representations?
The answer to this question is of great importance since the Oracle estimator is the cornerstone in every pursuit, and it provides an idealistic understanding of the capabilities of a given algorithm. In order to provide a complete picture, we first analyze the Oracle estimators for the layer-wise and projection approaches, and then proceed to analyze the holistic alternative proposed in this work. We note that the previous work \cite{papyan2016convolutional,sulam2017multi} on multi-layer sparse models have only considered bounded noise assumptions, adopting a worst-case point of view. Not only does this lead to very loose bounds, but it also blurs the real connection between the model features and the error these induce. Thus, we present a novel analysis of the Oracle estimator performance for all approaches under stochastic noise assumptions.

Let us start by recalling the average performance bounds for the general single-layer sparse model. Consider a signal $\y=\x+\e$, where $\x = \D \gama$, $\e\sim \mathcal{N}(\0,\sigma^2 \I)$ is a Gaussian noise, the representation true support is $\Lambda$ with cardinality $s$, and $\delta_s^{\D}$ is the RIP constant of the dictionary $\D$ \cite{candes2008restricted}. Then, the Oracle estimator is obtained via simple Least-Squares,
$
\hat{\gama}^{\Lambda} = {\D^{\Lambda}}^{\dagger} \y.
$
The above estimate can be equivalently expressed as
$
\hat{\gama}^{\Lambda} = \gama^{\Lambda} + \tilde{\e} = \gama^{\Lambda} + \sigma \left( 
	{\D^{\Lambda}}^T {\D^{\Lambda}}
    \right)^{-1/2} \z
$,
where $\tilde{\e} \sim \mathcal{N}(\0,\sigma^2 
	( 
	{\D^{\Lambda}}^T {\D^{\Lambda}}
    )^{-1})$, and $\z \sim \mathcal{N} (0,\I)$. 
Therefore, in expectation, one has that
\begin{equation}
\E \norm{\gama-\hat{\gama}}_2^2 = \sigma^2 \Trace 
	\left( 
	{\D^{\Lambda}}^T {\D^{\Lambda}}
    \right)^{-1},
\end{equation}
and the bounds on the recovery error can be shown to be (see \cite{candes2007dantzig,ben2010coherence}):
\begin{equation} \label{Oracle_bounds_one_layer}
\frac{\sigma^2 s}{1+\delta_s^{\D}} \leq \E \norm{\gama-\hat{\gama}}_2^2 \leq  \frac{\sigma^2 s}{1-\delta_s^{\D}}.
\end{equation}
As we see, the recovery error of the Oracle estimator is proportional to the representation cardinality. 

With these tools we now analyze the Oracle estimator performance for the different approaches for the multi-layer model. Consider a signal $\y = \x+\e$, but now $\x=\D_1 \gama_1 = \ldots = \D_{(k)}\gama_k $ is an ML-SC signal, the true support of layer $i$ is $\Lambda_i$ with cardinality $s_i$, and we denote by $\delta_{s_i}^{\D_i}$ the RIP constant of the dictionary $\D_i$ for cardinality $s_i$. 

In addition, we will need to define an appropriate extension of the RIP, which we named by Subset RIP, for those cases where not only the representation is sparse, but the signal is also sparse. Recall that if a matrix satisfies the RIP, the constant $\delta_s^{\D}$ provides a bound to the deviation of the singular-values of every sub-dictionary of $s$ columns from 1. The proposed extension provides a similar interpretation but for sub-dictionaries obtained by removing not just columns, corresponding to the support of a certain $\gama_i$, but also rows, corresponding to the support of $\gama_{i-1}$.

For normalized Gaussian random matrices, the singular values of a sub-dictionary are indeed expected to be centered at 1. If now certain $s$ out of $n$ rows are removed, one would expect the singular-values of those sub-dictionaries to be centered around $\frac{s}{n}$. Note that most matrices that satisfy the RIP (like sub-Gaussian matrices) would also satisfy the Subset RIP definition. We define the Subset RIP formally as follows.
\begin{definition} \label{rip_extension}
For any subset of $s_R$ rows, $\Lambda_R$, and any subset of $s_C$ columns, $\Lambda_C$, the matrix $\D \in \mathbb{R}^{n \times m}$ satisfies the Subset RIP with constant $\delta_{s_R,s_C}^{\D}$ if this is the minimum constant so that
\begin{equation} \label{rip_extension_definition}
\left( \frac{s_R}{n} - \delta_{s_R,s_C}^{\D} \right) \norm{\e }_2^2 \leq
\norm{\D^{\Lambda_R,\Lambda_C} \e }_2^2 
\leq
\left( \frac{s_R}{n} + \delta_{s_R,s_C}^{\D} \right) \norm{\e }_2^2
\end{equation}
holds for all vectors $\e$.
\end{definition}

The Subset RIP will become useful in the Oracle estimator analysis that we are about to present.

\subsection{Oracle Performance for Layer-Wise Pursuit}
In the layer-wise approach, we start the recovery process by estimating $\gama_1$, and obtaining $\hat{\gama}_1$. Then, we use $\hat{\gama}_1^{\Lambda_1}$ to estimate $\gama_2$, and so on to the deepest layer. In an oracle setting, the estimation of each layer is performed using the corresponding oracle estimators for each layer. One should wonder if the oracle estimation of $\gama_2$ should be carried out using the sub dictionary $\D_2^{\Lambda_2}$, or the more restricted version $\D_2^{\Lambda_1, \Lambda_2}$. As we will show towards the end of this section, the former is preferred, as the latter leads to a biased error. 

In this manner, we employ the zero extension of the previous layer, $\hat{\gama}_{i-1}$ in order to estimate $\gama_i$, which results in:
\begin{equation} \label{eq:layer_wise_oracle_estimator}
\begin{split}
\hat{\gama}_i^{\Lambda_i} &= {\D_i^{\Lambda_i}}^{\dagger} \hat{\gama}_{i-1}
= 
\left( {\D_i^{\Lambda_i}}^T {\D_i^{\Lambda_i}} \right)^{-1}
{\D_i^{\Lambda_{i-1},\Lambda_i}}^T \hat{\gama}_{i-1}^{\Lambda_{i-1}}
\\
& =
\left( {\D_i^{\Lambda_i}}^T {\D_i^{\Lambda_i}} \right)^{-1}
{\D_i^{\Lambda_{i-1},\Lambda_i}}^T 
\cdots
\left( {\D_1^{\Lambda_1}}^T {\D_1^{\Lambda_1}} \right)^{-1}
{\D_1^{\Lambda_1}}^T 
\y
 = \U_{(i,1)} \y.
\end{split}
\end{equation}
In this form, one can concisely express, for every $i^{th}$ layer:
$
\hat{\gama}_i^{\Lambda_i} = \gama_i^{\Lambda_i} + \U_{(i,1)}  \e
= \gama_i + \sigma \W_i \z 
$
where $\W_i = ( \U_{(i,1)} \U_{(i,1)}^T )^{1/2}$, and as before $\z \sim \mathcal{N} (0,\I)$. As we prove in \Cref{appendix_layer_wise}, the recovery error bounds are:
\begin{equation} \label{layer_wise_bounds}
\sigma^2 s_i
    \frac{1}{1+\delta_{s_1}^{\D_1}} 
    \prod_{j=2}^{i}
    \frac{\frac{s_{j-1}}{n_{j-1}} - \delta_{s_{j-1},s_j}^{\D_j} }
    {\left(1+\delta_{s_j}^{\D_j}\right)^2}
\leq
\E \norm{\gama_i-\hat{\gama}_i}_2^2  \\
\leq 
\sigma^2 s_i
    \frac{1}{1-\delta_{s_1}^{\D_1}} 
    \prod_{j=2}^{i}
    \frac{\frac{s_{j-1}}{n_{j-1}} + \delta_{s_{j-1},s_j}^{\D_j} }
    {\left(1-\delta_{s_j}^{\D_j}\right)^2}.
\end{equation}
These bounds show that, for a given layer, its oracle estimator error depends linearly with its sparsity level, i.e.: it is proportional to $\sigma^2 s_i$, just like the single-layer case.
It is worthwhile noting that the above constants at each layer depend on the particular setting of the model parameters, such as the ratio on non-zero elements in other layers. In the non-oracle case (where the supports are unknown) as shown in \cite{papyan2016convolutional,sulam2017multi}, these bounds become looser with the depth of the network.

\subsection{Oracle Performance for the Projection Pursuit}
In the projection approach, we start the recovery by using the effective model to estimate the deepest sparse representation, $\gama_k$:
$\hat{\gama}_k = \D_{(k)}^{\dagger} \y,$ for which we can provide the single-layer error bounds from \Cref{Oracle_bounds_one_layer}:
\begin{equation}
\frac{\sigma^2 {s_k}}{1+\delta_{s_k}^{\D_{(k)}}}
\leq
\E \norm{\gama_k-\hat{\gama}_k}_2^2
\leq 
\frac{\sigma^2 s_k}{1-\delta_{s_k}^{\D_{(k)}}}.
\end{equation}
In the next steps, we use the known mid-layers supports to back track $\hat{\gama}_k$ to shallower representations:
\begin{equation} \label{mid_layers_projection_connection}
\hat{\gama}_i = \D_{i+1}^{\Lambda_i, \Lambda_{i+1}} \hat{\gama}_{i+1} =
\D_{i+1}^{\Lambda_i, \Lambda_{i+1}}
\cdots
\D_{k}^{\Lambda_{k-1}, \Lambda_{k}} \hat{\gama}_k.
\end{equation}
As we prove in \Cref{appendix_projection}, this process results with the following mid-layers error bounds:
\begin{equation} \label{projection_bounds}
\sigma^2
s_k
\frac{c_{k_1}}{1+\delta_{s_k}^{\D_{(k)}}}
\leq
\E \norm{\gama_i-\hat{\gama}_i}_2^2 \\
\leq 
\sigma^2
s_k
\frac{c_{k_2}}{1-\delta_{s_k}^{\D_{(k)}}},
\end{equation}
where $c_{k_1} = \prod_{j=i+1}^{k} 
\left(
\frac{s_{j-1}}{m_{j-1}} - \delta_{s_{j-1},s_j}^{\D_{j}}
\right)$, and $c_{k_2} = \prod_{j=i+1}^{k} 
\left(
\frac{s_{j-1}}{m_{j-1}} + \delta_{s_{j-1},s_j}^{\D_{j}}
\right)$.

Interestingly, we can see that the recovery error of the intermediate layers no longer depends on their cardinality but rather on that of the deepest layer -- which typically results in a lower error. 
Just as in the layer-wise case, the particular constants depend on the model parameters. However, in this case, these values are given by the multiplication of terms from the deeper to shallower layers, as opposed to from shallower to deeper. In the non-oracle case, this fact causes the error to grow accordingly \cite{sulam2017multi}.

One might think that the deepest layer estimator in the projection method must be optimal, as it results from a global Least-Squares. We will show that this is in fact not the case, because the Least-Squares estimator is only optimal when no additional information can be exploited. 
As we present next, there is an alternative for estimating $\gama_k$ which explicitly exploits the additional supports information and leads to a significantly better error.

Before proceeding, we would like to demonstrate that employing the intermediate dictionaries, in their row-restricted versions, introduces a bias in the oracle estimators. For the sake of simplicity, consider a two-layer model, and separate the effective dictionary into two parts:
\begin{equation}
\y = \D_{(2)}^{\Lambda_2} \gama_2^{\Lambda_2} +\e = \D_1^{\Lambda_1} \D_2^{\Lambda_1,\Lambda_2} \gama_2^{\Lambda_2} +
\D_1^{\Lambda_1^c} \D_2^{\Lambda_1^c,\Lambda_2} \gama_2^{\Lambda_2} +\e.
\end{equation}
The above expression clearly shows that employing only $\D_1^{\Lambda_1} \D_2^{\Lambda_1,\Lambda_2} $ to   estimate $\gama_2$ simply ignores the second term, leading to a bias in the estimate. In this simple two-layer example, this bias can be expressed as
\begin{equation}
\b(\hat{\gama}_2) = 
\E \left[\hat{\gama}_2 \right] -\gama_2 
=
\left( \D_1^{\Lambda_1} \D_2^{\Lambda_1,\Lambda_2} \right)^{\dagger} \D_1^{\Lambda_1^c} \D_2^{\Lambda_1^c,\Lambda_2} \gama_2^{\Lambda_2},
\end{equation}
and generally, for $k$ layers, the bias becomes:
\begin{equation}
\b(\hat{\gama}_k) 
=
\left( \D_1^{\Lambda_1} \D_2^{\Lambda_1,\Lambda_2}\cdots \D_k^{\Lambda_{k-1},\Lambda_k} \right)^{\dagger} 
\left(
\D_{(k)}^{\Lambda_k}
-
\D_1^{\Lambda_1} \D_2^{\Lambda_1,\Lambda_2}\cdots \D_k^{\Lambda_{k-1},\Lambda_k}
\right)
\gama_2^{\Lambda_2}.
\end{equation}

\subsection{Oracle Performance for a Holistic Pursuit}
We have seen that the layer-wise and the projection approaches cannot be optimal as they both ignore some information. But how can one use all the model information simultaneously?
In this section we provide the answer to this question when the true supports are known. The solution is based on the synthesis-analysis understanding we have presented above, and the corresponding holistic approach. Analyzing the recovery error of this strategy will result in a significantly improved result that would confirm that the whole is more than merely the sum of its parts.

The synthesis-analysis dual interpretation provides a way to use the mid-layers supports when estimating the last layer. Indeed, as we have shown, $\gama_k$ does not lie in $\mathbb{R}^{s_k}$, but rather in the kernel of $\Phib^{\Lambda_k}$, which we define in \Cref{phi_matrix}. Therefore, the optimal Oracle estimator is:
\begin{equation}
\hat{\gama}_k^{\Lambda_k} = \underset{\gama_k^{\Lambda_k}}{\argmin}
\norm{\y - \D_{(k)}^{\Lambda_k} \gama_k^{\Lambda_k}}_2
~~ \st ~\gama_k^{\Lambda_k} \in \mathrm{ker} \{ \Phib^{\Lambda_k} \}.
\end{equation}
While this problem might look somewhat challenging, it admits a surprisingly simple solution. Let us define $\K$ to be an orthogonal matrix that spans the null space of $\Phib^{\Lambda_k}$.
Such a matrix can be obtained by computing the singular-value decomposition (SVD) of $\Phib^{\Lambda_k}$, and choosing the $s_k-r$ right-singular vectors corresponding to the zero singular values, where $r$ is the rank of $\Phib^{\Lambda_k}$. With it, we might rewrite the objective simply as:
\begin{equation} \label{Holistic_Oracle_K_matrix}
\hat{\alfa} = \underset{\alfa}{\argmin} \norm{\y - \D_{(k)}^{\Lambda_k} \K \alfa}_2,
\end{equation}
where $\alfa$ is of length $s_k-r$, and then choosing $\hat{\gama}_k^{\Lambda_k} = \K \hat{\alfa}$. The corresponding Oracle estimator for this problem (once more, given the support $\Lambda_k$), is given by
\begin{equation}
\hat{\gama}_k^{\Lambda_k} = \K \hat{\alfa} = \K \left( 
\D_{(k)}^{\Lambda_k} \K \right)^{\dagger} \y.
\end{equation}
Since the columns of $\K$ are orthonormal, the error in $\hat{\gama}_k^{\Lambda_k}$ is simple to analyze:
\begin{equation}
\E \norm{\gama_k-\hat{\gama}_k}_2^2 =
\E \norm{\K (\alfa - \hat{\alfa}) }_2^2
=\E \norm{\alfa - \hat{\alfa} }_2^2.
\end{equation}
The error in $\hat{\alfa}$ will be the single-layer Oracle error employed above in \Cref{Oracle_bounds_one_layer}, where the dictionary is given by $\D_{(k)}^{\Lambda_k} \K$. 
Using again the orthonormality of $\K$, one can bound the singular-values of $\D_{(k)}^{\Lambda_k} \K$ employing the RIP of $\D_{(k)}^{\Lambda_k}$,
\begin{align}
(1-\delta_{s_k}^{\D_{(k)}}) \norm{\alfa}_2^2 = &
(1-\delta_{s_k}^{\D_{(k)}}) \norm{\K \alfa}_2^2 
\leq \norm{\D_{(k)}^{\Lambda_k} \K \alfa}_2^2  \\
\leq & (1+\delta_{s_k}^{\D_{(k)}}) \norm{\K \alfa}_2^2 =
(1+\delta_{s_k}^{\D_{(k)}}) \norm{\alfa}_2^2.
\end{align}
This way, the the recovery error bounds for the Holistic Oracle estimator are\footnote{We omit the Holistic Oracle estimator performance proof since it is similar to the Projection recovery error bounds proof appears in \Cref{appendix_projection}}
\begin{equation}
\sigma^2 \left( s_k - r \right)
\frac{c_{k_1}}{1+\delta_{s_k}^{\D_{(k)}}}
\leq
\E \norm{\gama_i-\hat{\gama}_i}_2^2 
\leq 
\sigma^2 \left( s_k - r \right)
\frac{c_{k_2}}{1-\delta_{s_k}^{\D_{(k)}}},
\end{equation}
where $c_{k_1}$ and $c_{k_2}$ are the same as the constants defined for the Projection Oracle case.

As can be seen, the error is now proportional to the dimension of the kernel space of $\Phib^{\Lambda_k}$: $s_k - r$. This reveals the significant advantage of this multi-layer sparse construction. When employing this estimator, the recovery error decreases by a factor of $\frac{s_k - r}{s_k}$ compared to previous approaches. In addition, this demonstrates the crucial role of the matrix $\Phib^{\Lambda_k}$, as it determines to what extent the holistic version is better than the projection approach. For example, when the intermediate dictionaries are sparse and the non-zero coefficients of $\gama_k$ are randomly chosen (as done in \cite{sulam2017multi}), $\Phib^{\Lambda_k}$ is the zero matrix with high probability and its rank is zero. In such a case, the performance of the Holistic Oracle estimator is the same as the one obtained projection Oracle estimators. However, when the dictionaries are dense(r) and random, $\Phib^{\Lambda_k}$ has a full row rank, and thus $r = \sum_{i=1}^{k-1} \ell_i$, resulting in a significantly performance improvement.


\section{The Holistic Pursuit}
\label{Sec:HolisticPursuit}

Given the understanding of the benefits of exploiting the mid-layer co-supports, in this section we undertake the, perhaps, more interesting question: how can we design a pursuit algorithm that can implement these ideas in practice? 
Note that the estimation of the intermediate layers' zeros, or co-support, should be done with care, as their wrong estimation would cause a biased error by possibly projecting onto the wrong subspace.

In what follows we present an algorithm to estimate the intermediate co-supports and the corresponding matrix $\K$, resulting in a significant performance improvement. We name this approach the ``Holistic Pursuit'', as it gives the solution for the whole system simultaneously. For simplicity, we shall assume that one has access to the knowledge of the number of co-support elements at each layer, $\ell_i$. Withal, one could of course devise strategies in order to estimate these values if they are unknown, and we will comment on this towards the end of this section. We depict this algorithm in \Cref{Holistic_Pursuit_Algorithm}.

\begin{algorithm}
\caption{The Holistic Pursuit}
\label{Holistic_Pursuit_Algorithm}
\textbf{Input} 
\begin{itemize}
\item Signal $\y$ and dictionaries $\left\{ \D_i \right\}_{i=1}^k$.
\item ML-SC parameters: the mid-layer co-sparsity levels - $\left\{ \ell_i\right\}_{i=1}^{k-1}$ and the deepest layer sparsity - $s_k$.
\item The mid-layers minimum absolute value - $\left\{ \gama_i^{\min} \right\}_{i=1}^{k-1}$.
\end{itemize}
\textbf{Output} 
\begin{itemize}
\item Set of representations $\{ \hat{\gama}_i \}_{i=1}^k$.
\end{itemize}
\textbf{Initialization} 
\begin{itemize}
\item Initialize the mid-layer co-supports:
$\hat{\Lambda}^c_i = \emptyset ~ \forall 1 \leq i \leq k-1$.
\item Initialize the matrix that spans $\gama_k$'s subspace: $\K = \I_{m_k \times m_k}$.
\end{itemize}
\textbf{Holistic iterations:} preform the following steps for $\ell^{tot} = \sum_{i=1}^{k-1} \ell_i$ times in order to estimate $\K$:
\begin{enumerate}
\item Estimate $\hat{\gama}_k$ using the current estimation of $\K$ with \Cref{Constraint_Basis_Pursuit}:
\begin{equation}
\underset{\alfa,\gama_k}{\min}~ \frac{1}{2} \norm{\y - \D_{(k)} \K \alfa }_2^2 + \eta \norm{\gama_k}_1  \text{ s.t. } \gama_k = \K\alfa.
\end{equation}
\item Select a layer $g$ on which to add a co-support element, using \Cref{Holistic_Pursuit_Choose_Layer_g} (see below).
\item Find a new element $j$ which is the minimum absolute value in layer $g$: \\ $j \gets \argmin \abs{\D_{(g+1,k)}^{\hat{\Lambda}_g} \hat{\gama}_k}$.
\item Update the co-support estimation $\hat{\Lambda}^c_g$, and the corresponding $\K$: \\ $\K \gets \bigcup_{i=1}^{k-1} \mathrm{ker} \left\{ \D_{(i+1,k)}^{\hat{\Lambda}_i^c}  \right\}$.
\end{enumerate}
\textbf{Holistic final step:} Use the found subspace to estimate $\gama_k$:
\begin{enumerate}
\item Estimate $\hat{\gama}_k$ using the current $\K$, with \Cref{Constraint_Basis_Pursuit}.
\item Propagate $\hat{\gama}_k$ to the mid-layer representations: 
\begin{equation}
\hat{\gama}_i \gets \D_{(i+1,k)} \hat{\gama}_k,~ \forall~ 1 \leq i \leq k-1.
\end{equation}
\end{enumerate}
\end{algorithm}

\subsection{The Holistic Pursuit Algorithm}
The Holistic pursuit consists of two main steps, that are to be iterated. In the first step, one aims to estimate the sparse inner-most $\gama_k$ given the intermediate layers co-support elements that have been found in previous iterations -- initialized as the empty set. This estimation amounts to solving a constrained sparse coding problem, in essence searching for a sparse $\gama_k$ such that it is orthogonal to the rows of the previously found mid-layer co-support. In other words, we are interested in solving the following sub-problem:
\begin{equation} \label{eq:constrained_problem}
\underset{\gama_k}{\min} \frac{1}{2} \norm{\y - \D_{(k)}\gama_k}_2^2  + \eta \norm{\gama_k}_1 
~~ \st ~\gama_k \in \mathrm{ker} \{ \Phib \}.
\end{equation}
Note this problem is almost the same as the one solved by the projection approach from \cite{sulam2017multi}, except in the latter there is no constraint as to the subspace on which $\gama_k$ should live. Here, we are explicitly requiring that $\gama_k\in\ker\{\Phib\}$, therefore incorporating more information to help find $\gama_k$. Such information is still relevant in cases where the dictionaries (from layer 2 to $k$) are sparse, providing an advantage over the Layered and the Projection approaches in those cases as well.

The constrained problem in \Cref{eq:constrained_problem} is convex, and can be solved with a variety of methods. We propose to address it by employing the Alternating Direction Method of Multipliers (ADMM) \cite{boyd2011distributed}, for which we introduce a variable split. In addition, we enforce the subspace constraint by means of the matrix $\K$, which spans the $\ker \{ \Phib \}$, resulting in:
\begin{equation} \label{Constraint_Basis_Pursuit}
\underset{\alfa,\gama_k}{\min}~ \frac{1}{2} \norm{\y - \D_{(k)} \K \alfa }_2^2 + \eta \norm{\gama_k}_1  \text{ s.t. } \gama_k = \K\alfa.
\end{equation}
In order to minimize this new constrained problem, we construct the (normalized) augmented Lagrangian penalty by introducing the dual variable $\u$,
\begin{equation}
\underset{\alfa,\gama_k,\u}{\min}~ \frac{1}{2} \norm{\y - \D_{(k)} \K \alfa }_2^2	+ \eta \norm{\gama_k}_1 + \frac{\rho}{2} \norm{\K \alfa - \gama_k + \u}_2^2.
\end{equation}
This can now be minimized iteratively by alternating minimization with respect to $\alfa$, $\gama_k$ and the dual variable $\u$, and we detail this process in \Cref{ADMM_algorithm}. As can be seen, this minimization reduces to the iteration of simple operations: the problem in line (2) is solved in closed form by an entry-wise thresholding operator, while the step in line (3) reduces to a Least Squares estimate. Moreover the matrix that needs to be inverted for this step\footnote{This matrix is given by $\mathbf{H} = \left( (\D_{(k)}\K)^T \D_{(k)}\K \right)^{\dagger}$.} can be precomputed in advance, saving computations.

\begin{algorithm} [H]
\caption{ADMM for Constraint Lasso}
\label{ADMM_algorithm}
\begin{algorithmic}[1]
\State \textbf{while} not converged \textbf{do}
\State \hspace*{.25cm} $\gama_k \gets \underset{\gama}{\argmin}~ \eta \norm{\gama}_1 + \frac{\rho}{2} \norm{\K \alfa - \gama + \u}_2^2$~;
\State \hspace*{.25cm} $\alfa \gets \underset{\alfa}{\argmin} \frac{1}{2} \norm{\y - \D_{(k)} \K \alfa }_2^2 + \frac{\rho}{2} \norm{\K \alfa - \gama + \u}_2^2$~;
\State \hspace*{.25cm} $\u \gets \u + \rho (\K \alfa - \gama_k)$~;
\end{algorithmic}
\end{algorithm}


The second stage in the Holistic Pursuit is the estimation of the co-support from the intermediate layers, using the obtained $\hat{\gama}_k$ from the previous step. A number of options are available for such a process. One could attempt, for instance, to estimate the entire $\ell^{tot}$ co-support elements at once. However, this is prone to yield mistakes that would cause the projection of $\hat{\gama}_k$ onto an incorrect subspace. For this reason, we limit the search of these elements to one co-support element at a time. As there are multiple layers to choose from, the chosen layer should maximize the chances of obtaining a correct element of the co-support. In this spirit, one prefers a layer with a high value of $\gama_i^{\min}$ and with many co-support elements yet to be found, $\ell_i - \abs{\hat{\Lambda}_i^c}$. We propose choosing the layer by:
\begin{equation} \label{Holistic_Pursuit_Choose_Layer_g}
g \gets 
\underset{ i : ~ 
\abs{\hat{\Lambda}_i^c}< \ell_i
}{\argmax} ~~
\gama_i^{\min} 
\left(1 + \frac{1 + \mu^i_R \left(
\ell_i -\abs{\hat{\Lambda}_i^c} -1
\right)}{\sqrt{\ell_i - \abs{\hat{\Lambda}_i^c} }}\right)^{-1},
\end{equation}
where $\mu^i_R$ is the \emph{row mutual coherence} of dictionary $\D_{(i+1,k)}$:
\begin{equation} \label{eq:rows_mutual_coherence}
\mu_R \triangleq \max_{p\ne j} \abs{\d_p \d_j^T}.
\end{equation}
The choice of this particular expression will become clear later in the analysis of the algorithm. After choosing the layer $g$, the entry to be updated is chosen as the minimum absolute value of the inner products between $\gama_k$ and the rows of $\D_{(g,k)}$.

Once a new element has been found and added to the co-support, one must update the matrix $\K$ -- spanning a subspace that now includes the new added row -- which is to be used in the next iteration in the estimation of $\gama_k$. This way, the estimation of the inner-most representation becomes more and more accurate as the iterations proceed, and the algorithm continues until all the intermediate layers co-support elements are found. Finally, once all co-support rows in $\Phib$ have been identified, one estimates $\gama_k$ one last time by solving the problem in \Cref{Constraint_Basis_Pursuit}, and then computes the intermediate representations as: $\forall~ 1 \leq i \leq k-1 ~~ \hat{\gama}_i \gets \D_{(i+1,k)} \hat{\gama}_k.$

Before moving to the analysis of this algorithm, we would like to stress that the Holistic Pursuit suggests a general framework for search of sparse representations under this dual synthesis-analysis model. Indeed, while we have made each iteration specific, the general components of the algorithm can be changed to either more or less accurate steps. For example, one could estimate $\gama_k$ in a greedy way instead of employing a basis pursuit formulation. Alternatively, one might prefer to estimate blocks of the co-support elements of layer $g$ in each iteration, or even re-evaluate part of the co-support found and possibly replace elements, somewhat in the lines of the COSAMP \cite{needell2009cosamp} and Subspace Pursuit \cite{dai2009subspace} algorithms.

As stated in the beginning of this section, the presented algorithm assumes that the set of minimal entries, $\{ \gama_i^{\min} \}_{i=1}^{k-1}$, and the set of co-sparsity levels, $\{ \ell_i \}_{i=1}^{k-1}$, are assumed given. With small modifications, however, one could adapt the Holistic Pursuit to handle those cases in which this information is not available. For example, if the set minimal entries in the representations is unknown, one could remove $\gama_i^{\min}$ from the rule that determines which layer to address next -- essentially assuming that all layers use the same minimal value. On the other hand, if the set of the intermediate layer co-sparsity levels, $\{ \ell_i \}_{i=1}^{k-1}$, is not given, one could search for the minimum absolute value across all layers, or even search for the entry that has the smallest overall effect on the residual. This process should be iterated until some condition or threshold is met, such as having reached a total number of co-support elements, \emph{global co-sparsity} level $\ell^{tot}$, or a residual noise level.

\subsection{Theoretical Analysis}

In what follows, we demonstrate and theoretically analyze how the Holistic Pursuit leverages the constraints in the model, providing better estimates than previous approaches. We divide our derivations according to the two steps of the algorithm: the synthesis-pursuit, and the analysis-pursuit. 

Let us start by equivalently rewriting the objective function of the first step appears in \Cref{eq:constrained_problem}:
\begin{equation} \label{eq:constrained_lasso}
	\tilde{\gama}_k \gets \underset{\gama_k} {\argmin} \frac{1}{2}
    \|\y - \D_{(k)} \gama_k \|_2^2 + \eta \|\gama_k\|_1~\text{ s.t.}~\Phib \gama_k = \textbf{0}.
\end{equation}
This is an interesting constrained Lasso problem, that has been recently studied in \cite{james2013penalized}. We bring here the performance guarantees provided in that work, while allowing ourselves to omit tedious details that deviate from our main message. These can be easily found in \cite{james2013penalized}. 

\begin{lemma}{(Corollary 1 \cite{james2013penalized})}
Under mild assumptions on the dictionary $\D_{(k)}\in \mathbb{R}^{n \times m_k}$, on the constrained matrix $\Phib \in \mathbb{R}^{\ell \times m_k}$ and on the true representation $\gama_k\in \mathbb{R}^{m_k}$ (see \cite{james2013penalized}), if $\eta = (4\sqrt{2} -2)\sigma \sqrt{\frac{\log m_k}{n}}$, $\gama_k$ obeys the constraint $\Phib\gama_k = \0$, and $\y = \D_{(k)}\gama_k +\e$, where $\e \sim \mathcal{N}(0,\sigma^2 \I)$ is a random noise vector, then with probability at least $1-2/m_k$,
\begin{equation}
\norm{\tilde{\gama}_k - \gama_k}_2^2 \leq \max \{s_k - \ell~,~ \ell \} \frac{64\sigma^2 \log m_k}{\kappa_L^2 n}
,\end{equation}
where $\kappa_L$ is a constant related to the dictionary and the constrained matrix.
\end{lemma}

In comparison, solving the same problem in \Cref{eq:constrained_lasso} but without the constraint:
\begin{equation} \label{eq:unconstrained_lasso}
	\tilde{\gama}_k^{Unconst} \gets \underset{\gama_k} {\argmin} \frac{1}{2}
    \|\y - \D_{(k)} \gama_k \|_2^2 + \eta \|\gama_k\|_1,
\end{equation}
with $\eta = 4\sigma \sqrt{\frac{\log m_k}{n}}$ leads to the following recovery error bound \cite{james2013penalized}:
\begin{equation}
\norm{\tilde{\gama}_k^{Unconst} - \gama_k}_2^2 \leq s_k \frac{64\sigma^2 \log m_k}{\kappa_L^2 n}.
\end{equation}
As can be observed, the additional intermediate layer information reduces the recovery error from being proportional to $s_k$ to being proportional\footnote{We conjecture that a stronger claim could be formulated in terms of $s_k - \ell$, as opposed to $\max\{s_k-\ell, \ell \}$} to $\max \{s_k - \ell~,~ \ell \}$.
This is precisely the motivation behind the iterations in the Holistic Pursuit: the recovery error of $\gama_k$ is reduced at every iteration, leading to a higher probability to estimate a new co-support element from the intermediate layers.

The next lemma analyzes the second step of the algorithm: the pursuit of a new co-support element. This result is based on the derivations in \cite{peleg2013performance}. However, we restrict the analysis to the probability of finding one true co-support element, resulting in looser conditions. Its proof can be found in \Cref{app:HolisticAnalysis}.

\begin{lemma} \label{lemma_threshold}
Let $\hat{\gama}_k = \gama_k + \e$ be the estimation of the deepest layer, let $\hat{\Lambda}_i^c$ be the estimated co-support in the $i^{th}$ layer, where
$\sum_{i=1}^{k-1} \abs{ \hat{\Lambda}_i^c} = j-1$,
and let $\mu^i_R$ be the \emph{row-wise} mutual-coherence defined in \Cref{eq:rows_mutual_coherence}.
If
\begin{equation}
\norm{\e}_2
\leq
\max_{i : 
\abs{\hat{\Lambda}_i^c}< \ell_i
}
\gama_i^{\min} 
\left(1 + \frac{1 + \mu^i_R \left(
\ell_i -\abs{\hat{\Lambda}_i^c} -1
\right)}{\sqrt{\ell_i - \abs{\hat{\Lambda}_i^c} }}\right)^{-1}
\end{equation}
then the Holistic Pursuit algorithm succeeds in its $j^{th}$ iteration in recovering a new element from the mid-layers co-support.
\end{lemma}

We now combine both lemmas above in the form an overall theorem, providing a theoretical guarantee for the holistic pursuit. 
\begin{theorem} \label{Th:HolisticPursuit}
Consider an ML-SC signal $\x$ with intermediate layer co-sparsity levels $\{ \ell_i \}_{i=1}^{k-1}$ and a deepest layer sparsity of $s_k$, contaminated by random noise $\e \sim \mathcal{N}(0,\sigma^2 \I)$, resulting in the observation $\y=\x +\e$.
Under the mild assumptions appearing in Corollary 1 in \cite{james2013penalized}, and if for every $j^{th}$ iteration the following holds:
\begin{equation} \label{eq:main_theorem_condition}
\sqrt{\max \{s_k - j~,~ j \}}  \frac{8\sigma}{\kappa_L} \sqrt{\frac{ \log m_k}{n}}
\leq 
\max_{i : 
\abs{\hat{\Lambda}_i^c}< \ell_i
}
\gama_i^{\min} 
\left(1 + \frac{1 + \mu^i_R \left(
\ell_i -\abs{\hat{\Lambda}_i^c} -1
\right)}{\sqrt{\ell_i - \abs{\hat{\Lambda}_i^c} }}\right)^{-1},
\end{equation}
then, with probability exceeding $\left( 1-2/m_k \right)^{\ell^{tot}}$, where $\ell^{tot} = \sum_{i=1}^{k-1} \ell_i$,
the Holistic Pursuit succeeds in recovering the support of all sparse representations $\gama_i$, and the recovery error is bounded by
\begin{equation}
\norm{\hat{\gama}_k^{Holistic} - \gama_k}_2^2 
\leq \max \{s_k - \ell^{tot}~,~ \ell^{tot}\} \frac{64\sigma^2 \log m_k}{\kappa_L^2 n}.
\end{equation}
\end{theorem}

While the above Theorem does succeed in posing clear conditions for success of the Holistic Pursuit, the terms of success are somewhat disappointing. On the positive side, we count the fact that the error is shown to be proportional to $s_k-\ell$, which agrees with the oracle analysis from \Cref{Sec:Oracle}. On the negative side of the scales, we must mention that the probability for success that seems to be weak, and the condition in \Cref{eq:main_theorem_condition} is too convoluted and unclear. Still, in the spirit of this work, which aims to propose a first of its kind Holistic Pursuit for the ML-SC model, we find this Theorem encouraging. Further work should be invested, both in devising new such algorithms, and improving their theoretical study. 

\section{Numerical Results}
\label{Sec:Experiments}


In this section we present numerical results that demonstrate the Holistic Pursuit algorithm while comparing it with previous approaches. 
We consider a two layer sparse model with a signal dimension of $n=50$, layer dimensions of $m_1 =100$ and $m_2 = 50$ and with dictionaries that were sampled from a Gaussian distribution, $d_1(i,j)~\sim \mathcal{N}(0, \frac{1}{n})$ and $d_2(i,j)~\sim \mathcal{N}(0, \frac{1}{m_2})$.
We set the deepest layer sparsity to be $s_2$, and the desired co-sparsity in the first layer, $\ell_1$. 
We synthesize signals from this model by first randomly choosing the supports of $\gama_2$, and then multiplying by the corresponding matrix $\K$ (computed with the SVD decomposition of the matrix $\Phib$) with a random vector $\alfa$ sampled from a Gaussian distribution -- as explained in \Cref{Sec:Synthesis_Interpretation_Issues}.
Finally, we add white Gaussian noise to the signals, $\e \sim \mathcal{N}(0, \mathbf{I} \sigma^2)$, to obtain the measurement vector $\y$. We study the recovery error of $\gama_2$ as function of $\ell_1$, and present the results in \Cref{fig:Holistic_Pursuit_results_line}.

\begin{figure}
\centering
\begin{subfigure}{.42\textwidth}
\includegraphics[width = \textwidth]{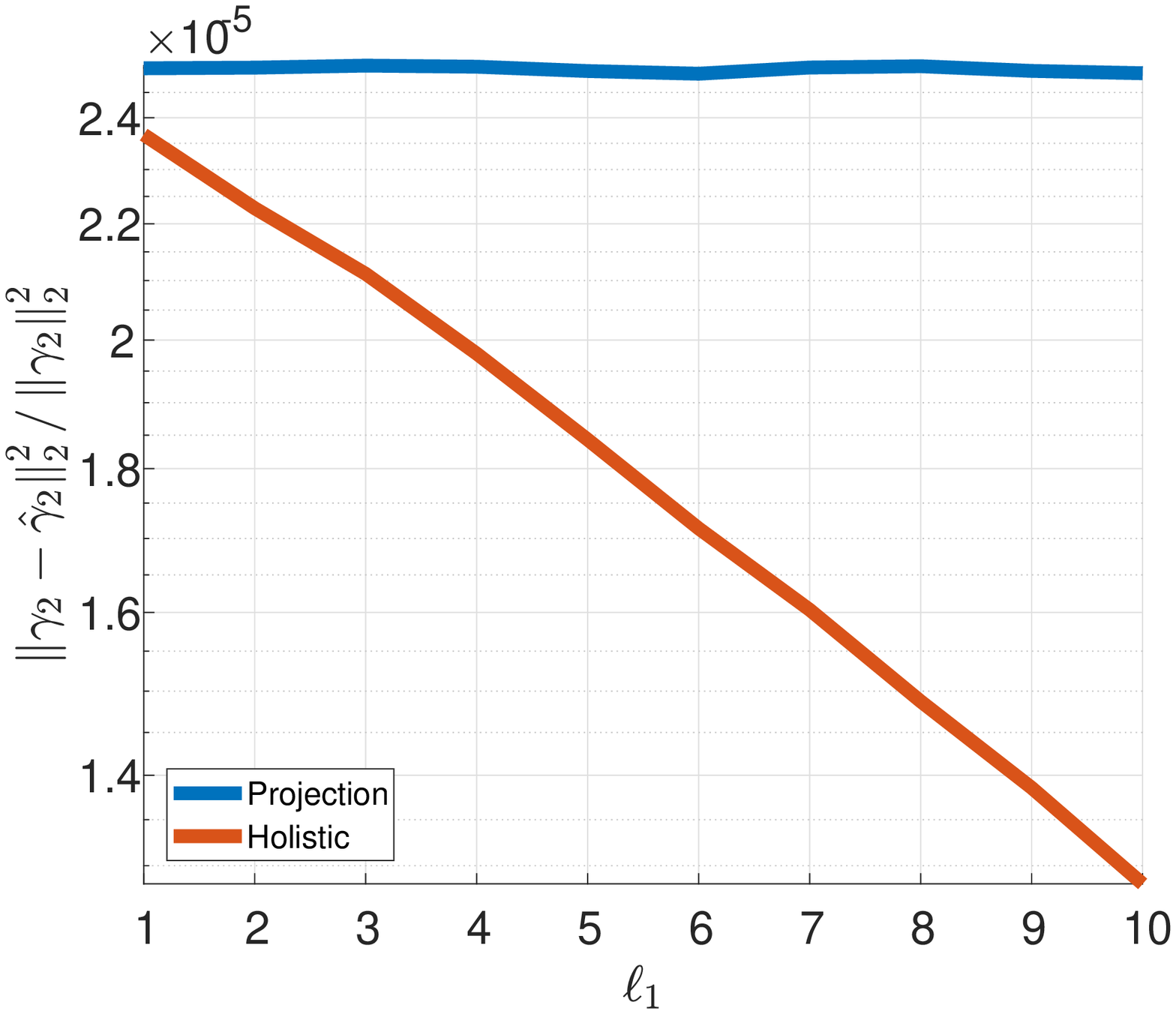}
\caption{Recovery error for the Holistic Pursuit algorithm and the outer-layer projection where the difference $s_2 - \ell_1 = 1$.}
\label{fig:Holistic_Pursuit_results_line}
\end{subfigure}
\hfil
\begin{subfigure}{.45\textwidth}
\includegraphics[width = \textwidth]{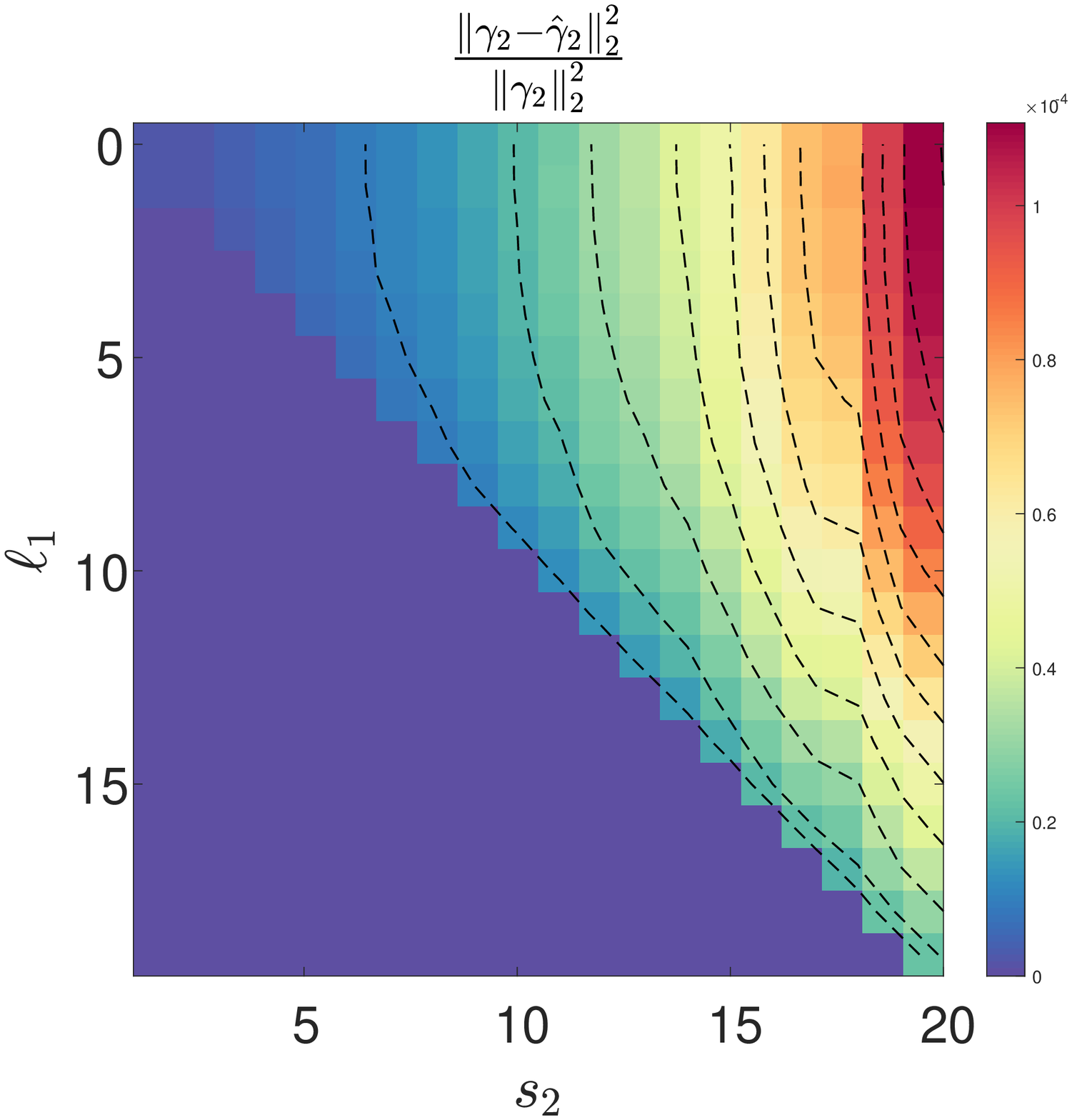}
\caption{Recovery error for the Holistic Pursuit as function of $\ell_1$ and $s_2$.}
\label{fig:Holistic_Pursuit_results_phase}
\end{subfigure}
\caption{Empirical recovery error for the Holistic Pursuit algorithm.}
\label{fig:NumericalResults}
\end{figure}

As discussed at length in  \Cref{Sec:Synthesis_Interpretation_Issues}, the layer-wise and projections approaches will not succeed in finding feasible estimates for $\gama_1$ and $\gama_2$.
The layer-by-layer algorithm will clearly fail, due to the very high cardinality of the first layer which is even bigger than the signal dimension, $s_1>n$. Moreover, $\gama_1$ is not the unique representation for $\x$ if one ignores the remaining layers.
The projection algorithm, on the other hand, would do somewhat of a better job in estimating $\gama_2$, alas it will obtain a dense estimate for $\gama_1$ when computing $\hat{\gama}_1 = \D_2\hat{\gama}_2$. The complete version of this algorithm, which attempts to correct the support in $\hat{\gama}_2$ if some constraints are not met, will eventually result in the zero-vector. 
This discussion exposes once more the fact that existing pursuit algorithms for the ML-SC model are suitable only for sparse dictionaries. All in all, we refrain from depicting the results from the layer-wise approach (as they do not provide competitive results), and we compare the Holistic Pursuit with the outer shell projection -- the main part in the projection pursuit (i.e., without backtracking and correcting its solution).

\Cref{fig:Holistic_Pursuit_results_line} presents the performance of both algorithms with a Signal to Noise Ratio (SNR) of 25 dB.
The penalty parameters, $\eta$, were set as the maximum value such that $ \norm{\y - \D_{(2)}\gama_2 }_2^2 \leq n \sigma^2 $.
One can see that, in both cases, as $\ell_1$ grows, the advantage of the Holistic Pursuit increases. This is to be expected, as the information of the co-support of $\gama_1$ becomes more significant and the effective dimension of the signal becomes smaller.

In \Cref{fig:HolisticPursuit_15db} we include further results for a SNR of 15 dB. Interestingly, the improvement of the Holistic Pursuit over the projection alternative is more significant in the high SNR scenario. For instance, when $\ell_1 = 10$ the Holistic algorithm reduces the error in 20\% in the 15 dB SNR case, and up to 50\% in the 25 dB SNR case.
 This behavior is explained by the fact that when the SNR is low, false detection in the estimation of the mid-layer co-supports are more likely, which in turn causes the selection of a wrong subspace on which to project $\hat{\gama}_2$. This also points to possible improvements for the proposed approach: in such cases, one should not try to estimate the complete co-support but rather we stop before its full recovery.

Before concluding, we depict in \Cref{fig:Holistic_Pursuit_results_phase} the recovery error for the Holistic Pursuit as a function of the enforced co-sparsity in the intermediate layer, $\ell_1$. In this case, signals were constructed exactly as described above, and the Holistic pursuit was run with different levels of the intermediate co-support. Recall that, as explained in \Cref{Sec:Synthesis_Interpretation_Issues}, in order to sample signals satisfying the model constraints one must require $s_2>\ell_1$. This figure clearly shows that, by explicitly leveraging the sparsity of $\gama_1$, our approach enables to recover denser representations with the same error as what the projection algorithm would have offered for a sparser signal. These iso-error curves are depicted in dashed black lines, which show the scaling of $s_2 - \ell_1$, as predicted by \Cref{Th:HolisticPursuit}.


\section{Conclusions}
\label{Sec:conclusions}
In this work we have revisited the multi-layer sparse model, and analyzed it in its most general (non-convolutional) form, providing the first known results for recovery of the model's sparse representations under random noise assumptions. The limitations of previous methods led us to propose a new interpretation of this multi-layer construction: this model can now be seen as a global synthesis construction with added analysis sparse priors. This understanding opened the door to both, a tighter theoretical analysis (such as uniqueness guarantees and oracle estimator performance) and the development of better pursuit algorithms to estimate the corresponding representations. This model has demonstrated, for the first time, the symbiotic effect of employing both synthesis and analysis priors on signals. While the Holistic Pursuit proposed above is a first implementation of these ideas, we envision several improvements to boost this algorithm. Naturally, future work should address the performance of this model on real data, for which a proper dictionary learning algorithm should be proposed. 

We have opted to consider the general sparse model and not dwell on the convolutional setting of \cite{sulam2017multi,papyan2016convolutional} in order keep the technical derivations simpler. While doing so, we have also considered the more general case of fully connected layers, as opposed to convolutional ones. While we see no fundamental problems that would prevent us from applying this dual synthesis-analysis interpretation to the convolutional setting, we also believe that this will call for a careful analysis and necessary adaptations. Not only it would be required to employ some of the tools already presented in previous works \cite{papyan2017working} (such as $\ell_{0,\infty}$ norm as a sparsity measure, shifted mutual coherence for dictionary characterization, etc) but others are likely to be needed as well. For instance, a convolutional ``$\ell_\infty$'' version of the co-sparsity measure would probably come into play. All these points indeed constitute interesting directions of future work. 

More broadly, the connections to deep neural networks has motivated much of our analysis. In this respect, just as the work in \cite{sulam2017multi,papyan2016convolutional} provided a multi-layer sparse model for convolutional neural networks, in this work we have presented and analyzed a multi-layer model for fully connected ones. More precisely, the forward pass in such fully connected networks can be seen as a pursuit algorithm for signals in the ML-SC model. This work goes further, however, as we have deepened the analysis of the this model, demonstrating that it is not empty and presenting a clear way to sample from it. In addition, we have shown (by deriving an oracle estimator) that neural networks can hope for far better performance in terms of the estimation of the representations if a holistic pursuit is carried out, delineating exciting prospects for the deep learning community. It is true, however, that the algorithm that we presented in order to exploit this synthesis-analysis interpretation does not speak the language of neural networks, as it is greedy in nature. Exploring how similar ideas can be implemented in terms of appropriate network architectures remains a promising and interesting open question. Indeed, some initial ideas have already appeared in the recent work \cite{sulam2018multiBP}, and we believe several others will follow.


\appendix

\section{Uniqueness Revisited}
\label{app:UniquenessRevProof}

In this section, we first re-state, and later prove improved uniqueness guarantees for the ML-SC model. For the following result, we will require the dictionaries to be in \emph{general position}, as in \cite{nam2013cosparse}, in the sense that for any set of supports, $\{\Lambda_i\}_{i=1}^k$, if there exists a vector $\alfa$ satisfying $\D_{(k)}^{\Lambda_k}\alfa=\0$ and $\Phib^{\Lambda_k}\alfa=\0$, then necessarily $\alfa=\0$. Moreover, and in order to avoid trivial solutions, one should require that $|\Lambda_k| \leq \rank\{\Phib^{\Lambda_k}\} + \rank\{\D^{\Lambda_k}_{(k)}\}$. In other words, a non-zero vector -- of particular size --  cannot be in both null-spaces simultaneously. Note that this is a fair requirement as it holds for almost all set of dictionaries in a Lebesgue measure sense.

\begin{theorem} 
Consider an ML-SC signal $\x$, and a set of dictionaries $\{\D_i\}_{i=1}^{k}$ in general position. 
If there exists a set of representations, $\{\gama_i \}_{i=1}^k$ satisfying
\begin{equation}
s_k \leq \frac{\sigma(\D_{(k)}) -1 }{2} + r,
\end{equation}
where $r = \rank\{ \Phib^{\Lambda_k} \}$ and $s_k$ is number of non-zero coefficients in the deepest layer, then this set is the \emph{unique} ML-SC representation for $\x$ such that its deepest layer has no more than $s_k$ non-zeros and the rank of the corresponding $\Phib^{\Lambda_k}$ is no greater than $r$.
\end{theorem}

\begin{proof}
Let us assume that $\{\gama_{i_a}\}_{i=1}^k$ and $\{\gama_{i_b}\}_{i=1}^k$ are two different representation sets for the ML-SC signal $\x$, such that 
\begin{equation} \label{eqapp:gamma_k_sparsity_condition}
\norm{\gama_{k_a}}_0 = \norm{\gama_{k_b}}_0 = s_k,
, \end{equation} 
and 
\begin{equation} \label{eqapp:Phi_rank_condition}
\rank \{\Phib^{\Lambda_{k_a}}_a\}\leq r,\quad \rank \{\Phib^{\Lambda_{k_b}}_b\}\leq r, 
\end{equation}
where $\Phib_a$ denotes the matrix $\Phib$ (as in \Cref{phi_matrix}) for representation set $a$, and similarly for $\Phib_b$.
In addition, we remind the reader that $\Phib^{\Lambda_{k}}$ restricts $\Phib$ to the support of $\gama_k$.
Define next the union of the supports $\Lambda_{k_a}$ and $\Lambda_{k_b}$, the deepest layer supports of $a$ and $b$ respectively, to be $\Lambda_{k_U}$, i.e.,
\begin{equation}
\Lambda_{k_U} \triangleq  \Lambda_{k_a} \cup \Lambda_{k_b}.
\end{equation}
Define $h_k$ to be the cardinality of the intersection of the deepest layer supports:
\begin{equation}
h_k \triangleq \abs{\Lambda_{k_a} \cap \Lambda_{k_b}},
\end{equation}
and define $u$ to be the cardinality of their union:
\begin{equation} \label{eqapp:Uniqueness_u_upper_bound}
u \triangleq \abs{ \Lambda_{k_U} } = 2s_k - h_k
. \end{equation}
We may now use the union sub-dictionary $\D_{(k)}^{\Lambda_{k_U}}$ to express:
\begin{equation}
\x = \D_{(k)}^{\Lambda_{k_U}} \gama_{k_a}^{\Lambda_{k_U}}
	= \D_{(k)}^{\Lambda_{k_U}} \gama_{k_b}^{\Lambda_{k_U}}
, \end{equation}
where the extra elements added to these supports are simply zeros. Following \cite{nam2013cosparse}, let us now define the sub-spaces where $\gama_{k_a}^{\Lambda_{k_U}}, \gama_{k_b}^{\Lambda_{k_U}}$ lie:
\begin{eqnarray} 
\mathcal{W}_a \triangleq \left\{ \alfa : \alfa \in \ker \{ \Phib_a^{\Lambda_{k_U}} \},~ Supp(\alfa) = \Lambda_{k_a} \right\}, \\
\mathcal{W}_b \triangleq \left\{ \alfa : \alfa \in \ker \{ \Phib_b^{\Lambda_{k_U}} \},~ Supp(\alfa) = \Lambda_{k_b} \right\}.
\end{eqnarray}

We now consider a difference vector $\boldsymbol\Delta = \gama_{k_a} - \gama_{k_b}$.
In what follows, we shall find a condition under which $\boldsymbol\Delta$ cannot be other than the zero vector, implying that the representation for $\x$ must be in fact unique.

The difference vector restricted to the union of support, $\boldsymbol\Delta^{\Lambda_{k_U}}$, must satisfy two conditions:
\begin{enumerate}
\item It must lie in the null space of $\D_{(k)}^{\Lambda_{k_U}}$, since
\begin{equation}
\D_{(k)}^{\Lambda_{k_U}} \boldsymbol\Delta^{\Lambda_{k_U}} = \0.
\end{equation}
\item It must lie in the union of $\mathcal{W}_a$ and $\mathcal{W}_b$\footnote{Since $\gama_{k_a} \in \mathcal{W}_a$ and $\gama_{k_b} \in \mathcal{W}_b$, then any linear combination of them must reside in $\left( \mathcal{W}_a + \mathcal{W}_b \right)$.}:
\begin{equation}
\boldsymbol\Delta^{\Lambda_{k_U}} \in \left( \mathcal{W}_a + \mathcal{W}_b \right).
\end{equation}
\end{enumerate}
Therefore, $\boldsymbol\Delta$ ought to be zero if 
\begin{equation} \label{eqapp:Uniqueness_Condition_d} 
  \left( \mathcal{W}_a + \mathcal{W}_b \right) \cap \ker \{ \D_{(k)}^{\Lambda_{k_U}} \} = \{\0\},
\end{equation}
for all $\{\gama_{i_a}\}_{i=1}^k, \{\gama_{i_b}\}_{i=1}^k$ satisfying \Cref{eqapp:gamma_k_sparsity_condition} and \Cref{eqapp:Phi_rank_condition}.

Following \cite{nam2013cosparse}, as the dictionaries are in general position, then as long as
\begin{equation} \label{eqapp:Uniqueness_Condition_d_different_cosupport}
\dim\{ \left(\mathcal{W}_a + \mathcal{W}_b \right) \} + \dim\{ \ker (\D_{(k)}^{\Lambda_{k_U}}) \} \leq u,
\end{equation}
\Cref{eqapp:Uniqueness_Condition_d} also holds. Note that the assumption about the dictionaries being in general position is a fair one, as it holds for almost every dictionaries set in Lebesgue measure \cite{nam2013cosparse}.


We shall now elaborate on this condition by upper bounding the two elements in the left term and developing a corresponding sufficient condition for uniqueness.
We start by injecting the effective dictionary spark to upper bound $\dim\{ \ker (\D_{(k)}^{\Lambda_{k_U}}) \}$. We know that the rank of $\D_{(k)}^{\Lambda_{k_U}}$ is no less than $\min\{u, \sigma(\D_{(k)})-1 \}$, because $u$ is the number of columns in this matrix, and every $\sigma(\D_{(k)})-1$ of its columns are linearly independent by definition of the spark. Recall that, due to the rank-nullity theorem, we have 
\begin{equation}
\dim \{ \ker ( \D_{(k)}^{\Lambda_{k_U}} ) \} = u - \rank\{ \D_{(k)}^{\Lambda_{k_U}} \},
\end{equation}
and so the dimension of the kernel of $\D_{(k)}^{\Lambda_{k_U}}$ is upper bounded by
\begin{equation} \label{eqapp:effective_dictionary_kernel_upper}
\dim\{ \ker (\D_{(k)}^{\Lambda_{k_U}} ) \} \leq u - \min\{u, \sigma(\D_{(k)})-1 \} 
=  \max \{0, u - \sigma(\D_{(k)})+1 \}.
\end{equation}
On the other hand, to upper bound $\dim\{ \left(\mathcal{W}_a + \mathcal{W}_b \right) \}$, we recall that $\dim (\mathcal{W}_a)$ and $\dim (\mathcal{W}_b)$ are less or equal to $s_k-r$ (from \Cref{eqapp:Phi_rank_condition}), resulting 
\begin{equation} 
\dim\{(\mathcal{W}_a + \mathcal{W}_b )\} \leq 2s_k - 2r = u + h_k - 2r.
\end{equation}
Therefore, the corresponding sufficient condition for \Cref{eqapp:Uniqueness_Condition_d_different_cosupport} is:
\begin{equation}
\max \{0, u - \sigma(\D_{(k)})+1 \} + u + h_k - 2r \leq u
.\end{equation}
Given the $\max$ in the above expression, let us analyze both cases separately. First, if $u \leq \sigma(\D_{(k)})-1$, then, by definition of the spark, $\boldsymbol\Delta$ must be zero in order to obtain 
\begin{equation}
\0 = \D_{(k)}^{\Lambda_{k_U}} \boldsymbol\Delta.
\end{equation}
As one can see, this boils down to the traditional condition for uniqueness of sparse (synthesis) representations -- namely, that $s_k<\sigma(\D_{(k)})/2$.

On the other end, when $u > \sigma(\D_{(k)})-1$, one obtains the condition:
\begin{equation}
u + h_k \leq \sigma(\D_{(k)}) + 2 r -1,
\end{equation}
which leads to
\begin{equation}
s_k \leq \frac{\sigma(\D_{(k)})-1}{2} + r.
\end{equation}
In other words, as long as this condition holds, $\boldsymbol\Delta = \0$, and so $\gama_{i_a} = \gama_{i_b}$ $\forall i$.
\end{proof}
Before concluding this section, note that the above result was derived for the case of $\|\gama_k\|_0 = s_k$ for simplicity, though the same results hold for every $\gama_k : \|\gama_k\|_0 \leq s_k$.

\section{The Oracle Estimator Performance Proof}

\subsection{Layer-Wise: Oracle Performance} \label{appendix_layer_wise}
\begin{proof}
Recalling from \Cref{eq:layer_wise_oracle_estimator} that the Oracle estimator for the layer-wise approach is:
\begin{equation}
\begin{split}
\hat{\gama}_i^{\Lambda_i} &= {\D_i^{\Lambda_i}}^{\dagger} \hat{\gama}_{i-1}
= 
\left( {\D_i^{\Lambda_i}}^T {\D_i^{\Lambda_i}} \right)^{-1}
{\D_i^{\Lambda_{i-1},\Lambda_i}}^T \hat{\gama}_{i-1}^{\Lambda_{i-1}}
\\
& =
\left( {\D_i^{\Lambda_i}}^T {\D_i^{\Lambda_i}} \right)^{-1}
{\D_i^{\Lambda_{i-1},\Lambda_i}}^T 
\cdots
\left( {\D_1^{\Lambda_1}}^T {\D_1^{\Lambda_1}} \right)^{-1}
{\D_1^{\Lambda_1}}^T 
\y
= \U_{(i,1)} \y,
\end{split}
\end{equation}
where 
\begin{equation}
\U_{(i,j)} \triangleq \left( {\D_i^{\Lambda_i}}^T {\D_i^{\Lambda_i}} \right)^{-1}
{\D_i^{\Lambda_{i-1},\Lambda_i}}^T 
\cdots \\
\left( {\D_j^{\Lambda_j}}^T {\D_j^{\Lambda_j}} \right)^{-1}
{\D_j^{\Lambda_{j-1},\Lambda_j}}^T,
\end{equation}
and $\Lambda_0 \triangleq \{1,\ldots,n\}$. Using the fact that $\hat{\gama}_i^{\Lambda_i} = \gama_i^{\Lambda_i} + \U_{(i,1)}  \e
= \gama_i + \sigma \W_i \z $, we can write:
\begin{equation} \label{eqapp:layer_wise_oracle_proof}
\begin{split}
\E \norm{\gama_i-\hat{\gama}_i}_2^2 
&= \E \norm{\sigma \W_i \z}_2^2 \\ 
&= \sigma^2 \Trace 
	\U_{(i,1)} \U_{(i,1)}^T \\
&= \sigma^2 \Trace 
	\U_{(i,2)} \left( {\D_1^{\Lambda_1}}^T {\D_1^{\Lambda_1}} \right)^{-1}  \U_{(i,2)}^T \\
&= \sigma^2 \Trace 
	 \left( {\D_1^{\Lambda_1}}^T {\D_1^{\Lambda_1}} \right)^{-1}  \U_{(i,2)}^T \U_{(i,2)} .\\
\end{split}
\end{equation}
We shell now use the fact that for every symmetric positive definite matrices $\A, \B$ the following holds:
\begin{equation} \label{eqapp:trace_property}
\Trace \A \B = \Trace \lambda^{\min}_{\A} \B + 
\Trace \left( \A - \lambda^{\min}_{\A} \I \right)\B
\geq
\lambda^{\min}_{\A} \Trace \B,
\end{equation}
where $\lambda^{\min}_{\A}$ is the minimal eigenvalue of $\A$. The inequality is true because $\A - \lambda^{\min}_{\A} \I$ is a semi-positive symmetric matrix, resulting that the matrix $\left( \A - \lambda^{\min}_{\A} \I \right)\B$ is a semi-positive symmetric matrix, and therefore, its trace, which equals to the eigenvalues sum, is bigger than 0. In our case, $\A$ is the matrix $\left( {\D_1^{\Lambda_1}}^T {\D_1^{\Lambda_1}} \right)^{-1}$, and $\B$ is the matrix $\U_{(i,2)}^T \U_{(i,2)}$. In addition, using the RIP of dictionary $\D_1$, we know that the eigenvalues of $\left( {\D_1^{\Lambda_1}}^T {\D_1^{\Lambda_1}} \right)^{-1}$ are no less than $\frac{1}{1+\delta_{s_1}^{\D_1}}$. Therefore, we can lower bound the recovery error, 
\begin{equation}
\begin{split}
& \E \norm{\gama_i-\hat{\gama}_i}_2^2 
 \geq \sigma^2 \frac{1}{1+\delta_{s_1}^{\D_1}} 
    \Trace \U_{(i,2)}^T \U_{(i,2)}
\\
& = 
	\sigma^2 \frac{1}{1+\delta_{s_1}^{\D_1}} 
    \Trace 
	\U_{(i,3)}^T  
    \left( {\D_2^{\Lambda_2}}^T {\D_2^{\Lambda_2}} \right)^{-1}
\cdot {\D_2^{\Lambda_{1},\Lambda_2}}^T 
	{\D_2^{\Lambda_{1},\Lambda_2}}
    \left( {\D_2^{\Lambda_2}}^T {\D_2^{\Lambda_2}} \right)^{-1}
    \U_{(i,3)}
     \\
& =
	\sigma^2 \frac{1}{1+\delta_{s_1}^{\D_1}} 
    \Trace 
{\D_2^{\Lambda_{1},\Lambda_2}}^T 
	{\D_2^{\Lambda_{1},\Lambda_2}}
    \left( {\D_2^{\Lambda_2}}^T {\D_2^{\Lambda_2}} \right)^{-1}
    \cdot \U_{(i,3)}
    \U_{(i,3)}^T  
    \left( {\D_2^{\Lambda_2}}^T {\D_2^{\Lambda_2}} \right)^{-1}.
\\
\end{split}
\end{equation}
Using again the fact which was introduce in  \Cref{eqapp:trace_property} and the Subset RIP of $\D_2^{\Lambda_{1},\Lambda_2}$, we can write:
\begin{equation}
\begin{split}
&  \E \norm{\gama_i-\hat{\gama}_i}_2^2 \\
& \geq 
	\sigma^2 \frac{1}{1+\delta_{s_1}^{\D_1}} 
    \left( \frac{s_1}{n_1} - \delta_{s_1,s_2}^{\D_2} \right)
    \Trace 
    \left( {\D_2^{\Lambda_2}}^T {\D_2^{\Lambda_2}} \right)^{-1}
    \cdot \U_{(i,3)}
    \U_{(i,3)}^T  
    \left( {\D_2^{\Lambda_2}}^T {\D_2^{\Lambda_2}} \right)^{-1}
\\
& \geq 
	\sigma^2 \frac{1}{1+\delta_{s_1}^{\D_1}} 
    \frac{\frac{s_1}{n_1} - \delta_{s_1,s_2}^{\D_2}}{\left(1+\delta_{s_2}^{\D_2}\right)^2}
    \Trace 
    \U_{(i,3)}
    \U_{(i,3)}^T  
\\
& \vdots
\\
& \geq 
	\sigma^2 \frac{1}{1+\delta_{s_1}^{\D_1}} 
    \prod_{j=2}^{i-1}
    \frac{\frac{s_{j-1}}{n_{j-1}} - \delta_{s_{j-1},s_j}^{\D_j} }
    {\left(1+\delta_{s_j}^{\D_j}\right)^2}
    \Trace 
    \U_{(i,i)}
    \U_{(i,i)}^T  
\\
& \geq 
	\sigma^2 \frac{1}{1+\delta_{s_1}^{\D_1}} 
    \prod_{j=2}^{i-1}
    \frac{\frac{s_{j-1}}{n_{j-1}} - \delta_{s_{j-1},s_j}^{\D_j} }
    {\left(1+\delta_{s_j}^{\D_j}\right)^2}
    \left(
    \frac{s_{i-1}}{n_{i-1}} - \delta_{s_{i-1},s_i}^{\D_i} 
    \right)
    \cdot \Trace 
    \left( {\D_i^{\Lambda_i}}^T {\D_i^{\Lambda_i}} \right)^{-2}
\\
& \geq 
	\sigma^2 s_i
    \frac{1}{1+\delta_{s_1}^{\D_1}} 
    \prod_{j=2}^{i}
    \frac{\frac{s_{j-1}}{n_{j-1}} - \delta_{s_{j-1},s_j}^{\D_j} }
    {\left(1+\delta_{s_j}^{\D_j}\right)^2} .
\\
\end{split}
\end{equation}

We omit the upper bound proof since it is equivalent to the above lower bound proof apart from changing signs ('-' to '+' and opposite).
\end{proof}

\subsection{Projection: Oracle Performance} \label{appendix_projection}
We next derive the performance bounds for the projection Oracle estimator \Cref{projection_bounds}.
\begin{proof}
The bound for the first step of the projection algorithm which is the deepest layer estimation, is the single-layer bound of the effective model. 
In order to bound the recovery error of the intermediate layers, $\{\gama_i\}_{i=1}^{k-1}$, we use the connection which was present in  \Cref{mid_layers_projection_connection}. For convenience sake we define
\begin{equation}
\overline{\D}_{i} = \D_{i}^{\Lambda_{i-1}, \Lambda_{i}}
\cdots
\D_{k}^{\Lambda_{k-1}, \Lambda_{k}} 
~~
\forall ~
2 \leq i \leq k .
\end{equation}
Therefore, we can write:
\begin{equation}
\begin{split}
\E \norm{\gama_i-\hat{\gama}_i}_2^2 
&= \E \norm{\sigma \overline{\D}_{i+1} 
{\D_{(k)}^{\Lambda_k}}^{\dagger} \z}_2^2 
\\ 
&= \sigma^2 \Trace 
    {{\D_{(k)}^{\Lambda_k}}^{\dagger} }^T
	{\overline{\D}_{i+1}}^T
    \overline{\D}_{i+1}
    {\D_{(k)}^{\Lambda_k}}^{\dagger} 
    \\
&= \sigma^2 \Trace 
	\left( {{\D_{(k)}^{\Lambda_k}}}^T 
    {\D_{(k)}^{\Lambda_k}} \right)^{-1}
	{\overline{\D}_{i+1}}^T
    \overline{\D}_{i+1} .
\\
\end{split}
\end{equation}
Using \Cref{eqapp:trace_property} and the RIP of the matrix $\D_{(k)}^{\Lambda_k}$, we can lower bound the recovery error, 
\begin{equation}
\begin{split}
\E \norm{\gama_i-\hat{\gama}_i}_2^2 
& \geq 
	\frac{1}{1+\delta_{s_k}^{\D_{(k)}}}
    \Trace 
    	{\overline{\D}_{i+1}}^T
    \overline{\D}_{i+1}
\\
& =
	\frac{1}{1+\delta_{s_k}^{\D_{(k)}}}
    \Trace 
    {\D_{i+1}^{\Lambda_i, \Lambda_{i+1}}}^T \D_{i+1}^{\Lambda_i, \Lambda_{i+1}}
    \\ & \qquad 
	\cdot \overline{\D}_{i+2}
    {\overline{\D}_{i+2}}^T .
    \\
\end{split}
\end{equation}
Using again the connection from  \Cref{eqapp:trace_property} and the Subset RIP of $\D_{i+1}^{\Lambda_i, \Lambda_{i+1}}$, we can write:
\begin{equation}
\begin{split}
\E \norm{\gama_i-\hat{\gama}_i}_2^2 
& \geq
	\frac{1}{1+\delta_{s_k}^{\D_{(k)}}}
    \left(
    \frac{s_i}{m_i}-\delta_{s_{i},s_{i+1}}^{\D_{i+1}}
    \right)
    \Trace 
	\overline{\D}_{i+2}^T
    {\overline{\D}_{i+2}}
\\
& \vdots
\\
& \geq 
	\sigma^2
\frac{\prod_{j=i+1}^{k-1} 
\left(
\frac{s_{j-1}}{m_{j-1}} - \delta_{s_{j-1},s_j}^{\D_{j}}
\right)}{1+\delta_{s_k}^{\D_{(k)}}}
    \cdot 
    \Trace 
	{\D_{k}^{\Lambda_{k-1}, \Lambda_{k}}}^T
    \D_{k}^{\Lambda_{k-1}, \Lambda_{k}}
\\
& \geq 
\sigma^2
s_k
\frac{\prod_{j=i+1}^{k} 
\left(
\frac{s_{j-1}}{m_{j-1}} - \delta_{s_{j-1},s_j}^{\D_{j}}
\right)}{1+\delta_{s_k}^{\D_{(k)}}}
\end{split}
\end{equation}

The upper bound proof is the same as the lower bound proof with changing signs ('-' to '+' and opposite).
\end{proof}

\section{Holistic Pursuit Algorithm Analysis}\label{app:HolisticAnalysis}

\begin{lemma} \label{eqapp:lemma_threshold}
Let $\hat{\gama}_k = \gama_k + \e$ be the estimation of the deepest layer, let $\hat{\Lambda}_i^c$ be the estimated co-support in the $i^{th}$ layer, where
\begin{equation}
\sum_{i=1}^{k-1} \abs{ \hat{\Lambda}_i^c} = j-1,
\end{equation}
and let $\mu^i_R$ be the \emph{row-wise} mutual-coherence defined in \Cref{eq:rows_mutual_coherence}.
If
\begin{equation}
\norm{\e}_2
\leq
\max_{i : 
\abs{\hat{\Lambda}_i^c}< \ell_i
}
\gama_i^{\min} 
\left(1 + \frac{1 + \mu^i_R \left(
\ell_i -\abs{\hat{\Lambda}_i^c} -1
\right)}{\sqrt{\ell_i - \abs{\hat{\Lambda}_i^c} }}\right)^{-1}
\end{equation}
then the Holistic Pursuit algorithm succeeds in its $j^{th}$ iteration in recovering a new element from the mid-layers co-support.
\end{lemma}

\begin{proof}
Let $i$ be 
\begin{equation}
\label{eqapp:lemma2_i_definition}
i =
\underset{ i : ~ 
\abs{\hat{\Lambda}_i^c}< \ell_i
}{\argmax} ~~
\gama_i^{\min} 
\left(1 + \frac{1 + \mu^i_R \left(
\ell_i -\abs{\hat{\Lambda}_i^c} -1
\right)}{\sqrt{\ell_i - \abs{\hat{\Lambda}_i^c} }}\right)^{-1} .
\end{equation}
The algorithm succeeds in its $j^{th}$ iteration if
\begin{equation} \label{eqapp:lemma2_success_condition}
\min \abs{\D_{(i+1,k)}^{\Lambda_i^c / \hat{\Lambda}_i^c} \hat{\gama}_k} < 
\min \abs{\D_{(i+1,k)}^{\Lambda_i} \hat{\gama}_k},
\end{equation}
where $\Lambda_i^c / \hat{\Lambda}_i^c$ is the set of the unfounded co-support elements in layer $i$.
Since the left term is only for co-support rows, we can simplify the left term:
\begin{equation} \label{eqapp:lemma2_success_condition_left_term}
\min \abs{\D_{(i+1,k)}^{\Lambda_i^c / \hat{\Lambda}_i^c} \hat{\gama}_k} = 
\min \abs{\D_{(i+1,k)}^{\Lambda_i^c / \hat{\Lambda}_i^c} \e}.
\end{equation}
In order to upper bound \Cref{eqapp:lemma2_success_condition_left_term}, we look for the error vector $\e$ that maximizes this term. Let us first assume that the rows in $\D_{(i+1,k)}^{\Lambda_i^c / \hat{\Lambda}_i^c}$ are orthonormal. In this simplified case, the error that maximizes \Cref{eqapp:lemma2_success_condition_left_term} is the vector obtained by the average distance to every row of $\D_{(i+1,k)}^{\Lambda_i^c / \hat{\Lambda}_i^c}$, having $\sqrt{\ell_i-|\hat{\Lambda}_i^c|}$ of these. In other words, we look for the error vector $\tilde{\e}$ such that 
\begin{equation} \label{eqapp:eq:error_maximize_inner_product}
\tilde{\e} = \frac{\norm{\e}_2}{ \sqrt{ \ell_i -\abs{\hat{\Lambda}_i^c} } }
(\D_{(i+1,k)}^{\Lambda_i^c / \hat{\Lambda}_i^c} )^T \1.
\end{equation}
Such an error vector leads to the following upper bound:
\begin{equation}
\min \abs{\D_{(i+1,k)}^{\Lambda_i^c / \hat{\Lambda}_i^c} \e}
\leq 
\min \abs{\D_{(i+1,k)}^{\Lambda_i^c / \hat{\Lambda}_i^c} \tilde{\e}} =
\frac{\norm{\e}_2}{ \sqrt[]{ \ell_i -\abs{\hat{\Lambda}_i^c} } }.
\end{equation}

Consider now the more general case, where the rows of $\D_{(i+1,k)}^{\Lambda_i^c / \hat{\Lambda}_i^c}$ are not orthogonal, but rather the correlation between every two rows is upper bounded by $\mu^i_R$.
Now, in order to bound the minimal correlation between the noise vector and every atom, we might consider the worst-case scenario where the inner product between any pair of atoms is equal to $\mu^i_R$. In such a case, $\e$ in \Cref{eqapp:eq:error_maximize_inner_product} maximizes \Cref{eqapp:lemma2_success_condition_left_term}, and one thus obtains the following upper bound:
\begin{equation}\label{eqapp:lemma2_right_term_bound}
\min \abs{\D_{(i+1,k)}^{\Lambda_i^c / \hat{\Lambda}_i^c} \e}
\leq \norm{\e}_2
\frac{ \left(1 + \mu^i_R \left(
\ell_i -\abs{\hat{\Lambda}_i^c} -1
\right)
\right)
}{ \sqrt{ \ell_i -\abs{\hat{\Lambda}_i^c} } }.
\end{equation}

For the right term of \Cref{eqapp:lemma2_success_condition} we use  the same derivations as in \cite{peleg2013performance}, resulting in
\begin{equation}
\label{eqapp:lemma2_left_term_bound}
\min \abs{\D_{(i+1,k)}^{\Lambda_i} \hat{\gama}_k}
\geq
\gama_i^{\min} - 
\max \abs{\D_{(i+1,k)}^{\Lambda_i} \e}
\geq
\gama_i^{\min} - 
\norm{\e}_2,
\end{equation}
where the last inequality follows from Cauchy-Schwarz inequality and the fact that all rows in $\D_{(i+1,k)}$ are assumed to be normalized.
Using \Cref{eqapp:lemma2_right_term_bound} and \Cref{eqapp:lemma2_left_term_bound}, one arrives to a sufficient condition for the success of the algorithm, in the form of:
\begin{equation}
\norm{\e}_2
\leq
\gama_i^{\min} 
\left(1 + \frac{1 + \mu^i_R \left(
\ell_i -\abs{\hat{\Lambda}_i^c} -1
\right)}{\sqrt{\ell_i - \abs{\hat{\Lambda}_i^c} }}\right)^{-1}.
\end{equation}
Bringing \Cref{eqapp:lemma2_i_definition} provides the claimed lemma.
\end{proof}

\section{Numerical Experiments}

In this section, we expand on the experimental results presented in \Cref{Sec:Experiments}

\begin{figure}[h!]
\centering
\includegraphics[width = .6\textwidth]{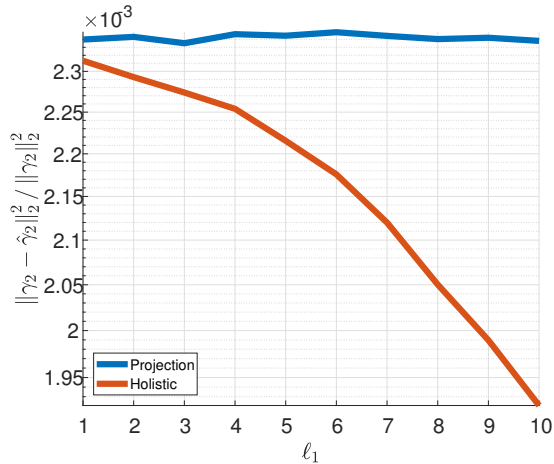}
\caption{Recovery error for the Holistic Pursuit algorithm and the outer-layer projection where the difference $s_2 - \ell_1 = 1$.}
\label{fig:HolisticPursuit_15db}
\end{figure}

\bibliographystyle{siamplain}
\bibliography{bibliography}
\end{document}